\documentclass[a4paper,UKenglish,cleveref, autoref, thm-restate]{lipics-v2021-2}

\pdfoutput=1 
\hideLIPIcs  

\usepackage{booktabs}       
\usepackage{amsfonts}       
\usepackage{nicefrac}       
\usepackage{microtype}      
\usepackage{lipsum}
\usepackage{graphicx}       
\graphicspath{{media/}}     
\usepackage{todonotes}
\usepackage{algorithm}
\usepackage[noend]{algpseudocode}
\usepackage{stackengine}
\usepackage{multirow,bigdelim}
\usepackage{colonequals}

\mathchardef\mhyphen="2D 
\usepackage{arydshln}
\usepackage{xcolor}
\usepackage{graphicx}
\usepackage{bigstrut}
\usepackage{complexity}
\usepackage{comment}
\usepackage{tikz}
\usepackage{adjustbox}
\usetikzlibrary{positioning} 
\usepackage{amssymb,amsmath,amsthm}
\usepackage{mathtools}
\usepackage{braket}

\usepackage{algorithmicx,algorithm,eqparbox,array}

\newcommand{\vect}[1]{\boldsymbol{#1}}

\newcommand*{\logeq}{\ratio\Leftrightarrow}

\newcommand{\frakA}{\mathfrak{A}}
\newcommand{\calA}{\mathcal{A}}
\newcommand{\calB}{\mathcal{B}}

\newcommand{\frakI}{\mathfrak{I}}

\newcommand{\calL}{{\mathcal{L}}}

\newcommand{\calM}{\mathcal{M}}

\newcommand{\aent}[1]{{\overline{\Gamma^*_{#1}}}}

\newcommand{\step}[1]{\mathrm{S}_{#1}}

\newcommand{\tup}{\mathrm{Tup}}

\newcommand{\Supp}{\mathrm{Supp}}
\newcommand{\att}{\mathsf{Var}}
\newcommand{\val}{\mathsf{Val}}

\newcommand{\Rpos}{\mathbb{R}_{\geq 0}}
\newcommand{\ar}{\mathrm{ar}}

\newcommand{\pind}{\perp\!\!\!\perp}

\newcommand {\pci}[3] {#2~\!\!\pind\!\!~#3 \mid #1}
\newcommand {\pmi}[2] {#1~\!\!\pind\!\!~#2}

\newcommand{\join}[2]{#1 \!\bowtie \!#2}

\newcommand{\kplus}{\oplus}
\newcommand{\ktimes}{\otimes}
\newcommand{\bigkplus}{\bigoplus}
\newcommand{\bigktimes}{\bigotimes}
\newcommand{\Dom}{\mathrm{Dom}}

\makeatletter
\def\BState{\State\hskip-\ALG@thistlm}
\makeatother

\title{Conditional independence on semiring relations} 

\author{Miika Hannula\footnote{The author has been supported by the ERC grant 101020762.}}{Department of Mathematics and Statistics\\ University of Helsinki\\Finland\\\texttt{miika.hannula@helsinki.fi}}{}{}{}

\authorrunning{M. Hannula} 

\Copyright{Jane Open Access and Joan R. Public} 

\category{} 

\relatedversion{} 

\nolinenumbers 

\begin{document}

\maketitle

\begin{abstract}
Conditional independence plays a foundational role in database theory, probability theory, information theory, and graphical models. In databases, a notion similar to conditional independence, known as the (embedded) multivalued dependency, appears in database normalization. Many properties of conditional independence are shared across various domains, and to some extent these commonalities can be studied through a measure-theoretic approach.
The present paper proposes an alternative approach via semiring relations, defined by extending database relations with tuple annotations from some commutative semiring. Integrating various interpretations of conditional independence in this context, we investigate how the choice of the underlying semiring impacts the corresponding axiomatic and decomposition properties.
We specifically identify positivity and multiplicative cancellativity as the key semiring properties that enable extending results from the relational context to the broader semiring framework. Additionally, we explore the relationships between different conditional independence notions through model theory, and consider how methods to test logical consequence and validity generalize from database theory and information theory to semiring relations.
\end{abstract}

\section{Introduction}
Conditional independence (CI) is an expression of the form $\pci{X}{Y}{Z}$, stating that $Y$ and $Z$ are conditionally independent given $X$. Common to its different interpretations is that conditional independence is a mark of redundancy. For instance, on a relation schema over attributes $X,Y,Z$, the multivalued dependency (MVD) $X \twoheadrightarrow Y$ can be viewed as the counterpart of the CI $\pci{X}{Y}{Z}$, expressing that a relation can be losslessly decomposed into its projections on $X,Y$ and $X,Z$. The process of splitting the schema into smaller parts---in order to avoid data redundancy---is called normalization, and a database schema is in fourth normal form if every non-trivial MVD follows from some key. In probability theory, CIs over random variables give rise to factorizations of joint probability distributions into conditional distributions. Since the decomposed distributions can be represented more compactly, this allows more efficient reasoning about the random variables. In addition to these classical examples, conditional independence has applications in ordinal conditional functions \cite{Spohn1988}, Dempster-Schaefer theory \cite{dempster,shafer1976mathematical}, and possibility theory \cite{ZADEH19783}. 
 
 Since the notion of conditional independence has a relatively fixed meaning across various contexts,
 it is no coincidence that the central rules governing its behavior are universally shared.
 The semigraphoid axioms \cite{Pearl88a} state five basic rules that hold true for diverse interpretations of conditional independence.  Initially conjectured to be complete by Pearl, Studen\'{y} \cite{Studeny89} proved incompleteness of these rules by discovering a 
 new rule
that is not derivable by the semigraphoid axioms, while being sound for probability distributions. Later he \cite{studeny:1993} proved that there cannot be any finite axiomatization for conditional independence, a fact that had been established earlier for embedded multivalued dependencies (EMVDs) \cite{parker:1980}.
The \emph{implication problem}, which is to determine whether some set of dependencies $\Sigma$ logically implies a dependency $\tau$,  is in fact undecidable not only for EMVDs \cite{herrmann95}, but also for CIs in probability theory, as has been recently shown \cite{kuhne22,Li23}. 
 
 In some partial cases, the semigraphoid axioms are known to be complete. 
 A \emph{saturated conditional independence} (SCI) is a CI that contains all the variables of the underlying joint distribution.
 The semigraphoid axioms 
 are complete for the implication of arbitrary CIs by saturated ones under various semantics \cite{GyssensNG14},
 and for the implication of CIs from a set of CIs encoded in the topology of a Bayesian network \cite{GeigerVP90}.
 In databases, where SCIs correspond to MVDs, the implication problem for MVDs combined with functional dependencies (FDs) is 
 well-known to have a finite axiomatization and a polynomial-time algorithm  \cite{BeeriFH77}.

 Moving beyond saturated CIs, the implication problem not only becomes undecidable, but also more sensitive to the underlying semantics. Studen\'{y} \cite{Studeny93} presents several example inference rules that involve non-saturated CIs and are sound in one setting while failing to be sound in others.  For instance, the aforementioned rule\footnote{This rule states that $\pci{CD}{A}{B} \land \pci{A}{C}{D}\land \pci{B}{C}{D} \land \pci{\emptyset}{A}{B} $ if and only if $ \pci{AB}{C}{D} \land \pci{C}{A}{B}\land \pci{D}{A}{B} \land \pci{\emptyset}{C}{D}$. For probability distributions the rule follows by the non-negativity of conditional mutual information $I(Y;Z|C)$, and the fact that  $I(Y;Z|X)=0$ if and only if $Y$ and $Z$ are conditionally independent given $X$. For database relations the rule is not sound; see a counterexample in \cite{Studeny93}.} showing incompleteness of the semigraphoid axioms is not sound for database relations, but its
   soundness for probability distributions
 follows by a simple information-theoretic argument.
FDs and MVDs can also be alternatively expressed in terms of information measures over a uniformly distributed database relation \cite{Lee87}, and their implication problem has recently been connected to validity of information inequalities \cite{KenigS22}. Galliani and V\"{a}\"{a}n\"{a}nen \cite{GallianiV22} associate relations with a so-called diversity measure to capture FDs and other data dependencies. These measure-theoretic approaches, however, fail to capture the semantics of the embedded multivalued dependency in full generality.

This paper examines $K$-relations as a unifying framework for conditional independence and other dependency concepts.  Introduced in the seminal work \cite{GreenKT07}, $K$-relations extend ordinary relations by tuple annotations from a commutative semiring $K$, providing a powerful abstraction for data provenance.
While it is natural to consider propagation of tuple annotations through queries in this context, one can also ask how tuple annotations couple with data dependencies. Dependencies on $K$-relations have thus far received limited attention (see, e.g., \cite{Atserias20,BarlagHKPV23,ChuMRCS18}). Related to this work, Barlag et al. \cite{BarlagHKPV23} define conditional independence for $K$-relations, and raise the question of how much the related axiomatic properties depend on the algebraic properties of $K$. Atserias and Kolaitis \cite{Atserias20} study the relationship between local and global consistency for $K$-relations, introducing also many concepts that will be adopted in this paper. Although the authors do not  consider conditional independence, they show that functional dependencies on $K$-relations entail lossless-join decompositions.

The following contributions are presented in this paper: First, we show that conditional independence for $K$-relations corresponds to lossless-join decompositions whenever $K$ is positive and multiplicatively cancellative. 
Then, we provide a proof that, for any $K$ exhibiting these characteristics, the semigraphoid axioms are sound for general CIs, and extend to a complete axiomatization of SCI+FD which is comparable to that of MVD+FD. This entails that database normalization techniques extend to $K$-relations whenever positivity and multiplicative cancellativity are assumed. To showcase potential applications, we illustrate through an example how the semiring perspective can lead to decompositions of data tables which appear non-decomposable when interpreted relationally. We also explore how $K$-relations and model theory can shed light into the interconnections among different CI semantics. Lastly, we compare two standard methods used in database theory and information theory to test logical consequence and validity for constraints: the chase and the copy lemma. In particular, we demonstrate that the chase can sometimes be simulated using the copy lemma, and prove that the latter method extends to $K$-relations whenever $K$ is positive.

\section{Semirings}
We commence by recapitulating concepts related to semirings.
  A \emph{semiring} is a tuple $K=(K,\kplus,\ktimes,0,1)$, where $\kplus$ and $\ktimes$ are binary operations on $K$, $(K,\kplus,0)$ is a commutative monoid with identity element $0$, $(K,\ktimes ,1)$ is a monoid with identity element $1$, $\ktimes$ left and right distributes over $\kplus$, and $x \ktimes 0 =0= 0\ktimes x$ for all $x \in K$.
  The semiring $K$ is called \emph{commutative} if $(K,\ktimes ,1)$ is a commutative monoid.  That is, semirings are rings which need not have additive inverses.
   As usual, we often
  write $ab$ instead of $a \ktimes b$. In this paper, we assume that every semiring is non-trivial ($0\neq 1$) and commutative.
  The symbols $\kplus, \ktimes,\bigkplus, \bigktimes$ are used in reference to specific semiring operations, and symbols $+,\cdot,\sum,\prod$ refer to ordinary arithmetic operations.
  
We list some example semirings that will be considered in this paper. 
  \begin{itemize}
    \item The \emph{Boolean semiring} $\mathbb{B}=(\mathbb{B},\lor,\land,0,1)$ models logical truth and is formed from the two-element Boolean algebra. It is the simplest example of a semiring that is not a ring.
    \item The \emph{probability semiring} $\mathbb{R}_{\geq 0}=(\mathbb{R}_{\geq 0},+,\cdot,0,1)$ consists of the non-negative reals with standard addition and multiplication.
    \item The \emph{semiring of natural numbers} $\mathbb{N}=(\mathbb{N},+,\cdot,0,1)$ consists of natural numbers with their usual operations.
    \item The \emph{tropical semiring} $\mathbb T= (\mathbb{R}\cup\{\infty\}, \min, +, \infty, 0)$ consists of the reals expanded with infinity
	and has $\min$ and standard addition respectively plugged in for addition and multiplication.
\item The \emph{Viterbi semiring} $\mathbb V = ([0,1], \max, \cdot, 0,1)$ associates the unit interval with maximum as addition and standard multiplication. 
   \end{itemize}
   Other examples include the semiring of multivariate polynomials $\mathbb{N}[\vect{X}]=(\mathbb{N}[\vect{X}],+,\cdot,0,1)$ which is the free commutative semirings generated by the indeterminates in $\vect{X}$,  and
%
the \emph{Lukasiewicz semiring} $\mathbb{L} = ([0,1], \max, \cdot , 0, 1)$, used in multivalued logic, which endows the 
	unit interval with $\max$ addition and multiplication $a \cdot b \coloneqq \max(0, a+b-1)$.

  Let $\leq$ be a partial order. 
  A binary operator $*$ is said to be \emph{monotone under $\leq$} if $a\leq b$ and $a'\leq b'$ implies $a*a' \leq b*b'$. 
  If $*=\kplus$ (resp. $*=\ktimes$), we call this property of $(K,\leq)$ \emph{additive monotony} (resp. \emph{multiplicative monotony}).
  A \emph{partially ordered semiring} is a tuple $K=(K,\kplus,\ktimes,0,1, \leq)$, where $(K,\kplus,\ktimes,0,1)$ is a semiring, and $(K,\leq)$ is a partially ordered set satisfying additive and multiplicative monotony. Given a semiring ${K}=(K,\kplus,\ktimes,0,1)$, define a binary relation $\leq_{K}$ on $K$ as 
  \begin{equation}\label{eq:nat}
  a \leq_{K} b\logeq\exists c : a\kplus c = b.
 \end{equation}
  This relation is a preorder, meaning it is reflexive and transitive. If $\leq_{K}$ is also antisymmetric, it is a partial order, called the \emph{natural order} of ${K}$, and ${K}$ is said to be \emph{naturally ordered}. 
  In this case, ${K}$ endowed with its natural order is a partially ordered semiring. If additionally the natural order is \emph{total}, i.e.,  $a\leq_K b$ or $b\leq_K a$ for all $a,b \in K$, we say that $K$ is \emph{naturally totally ordered}.

  If a semiring $K$ satisfies $ab=0$ for some $a,b\in K$ where $a\neq 0 \neq b$, 
  we say that $K$ has \emph{divisors of $0$}. The 
  semiring $K$ is called \emph{$\kplus$-positive} if $a\kplus b=0$ implies 
  that $a=b=0$. If $K$ is both $\kplus$-positive and has no divisors of $0$, 
  it is called \emph{positive}.
  For example, the modulo two integer semiring $\mathbb{{Z}}_2$ 
  is not positive since it is not $\kplus$-positive (even though it has no divisors of $0$). 
  Conversely, an example of a semiring with divisors of $0$ is $\mathbb{{Z}}_4$.
  A semiring is called 
  \emph{additively} (resp. \emph{multiplicatively}) \emph{cancellative}
  if $a \kplus b = a \kplus c$  implies $b=c$
  (resp.  $ab = ac$ and $a \neq 0$ implies $b=c$). It is simply
  \emph{cancellative} if it is both additively and multiplicatively
  cancellative. 
A semiring $K$ in which each non-zero element has a multiplicative inverse is called a \emph{semifield}. A semifield $K$ in which each element has an additive inverse is a \emph{field}.

  In particular, note that
	the probability semiring  $\mathbb{R}_{\geq 0}$, the semiring of natural numbers $\mathbb{N}$, the Boolean semiring $\mathbb{B}$, and the tropical semiring are positive, multiplicatively cancellative, and naturally ordered.
  Of these only the first two are also additively cancellative. This difference seems to be crucial for the behavior of conditional independence.

 This section concludes with two lemmata. The first lemma is applied when examining the relationship between lossless-join decompositions and conditional independence (Theorem \ref{thm:decomp}). The second lemma comes into play when comparing the CI implication problem for different semirings (Theorem \ref{thm:rpostoK}).  
 A formal definition of an \emph{embedding} of a model into another model and the lemma proofs are located in Appendix \ref{sect:embeddings}. 
 A field $F$ endowed with a total order $\leq$ is a \emph{totally ordered field} if $(F,\leq)$ satisfies additive monotony and \emph{monotony of non-negative multiplication}: $a \geq 0$ and $ b\geq 0$ implies $ab\geq 0$.

\begin{restatable}{lemma}{possemi}\label{prop:embedding}
  Any positive multiplicatively cancellative semiring $K$ embeds in a positive semifield $F$.
  Furthermore, 
  if $K$ is additively cancellative, then $F$ is additively cancellative, and
  if $K$ is a naturally totally ordered, then $F$ is naturally totally ordered.
\end{restatable}

 \begin{lemma}\label{lem:semi-cons}
Any naturally totally ordered cancellative semiring
  embeds in a totally ordered field.
  \end{lemma}
  
  \section{$K$-relations}

This section introduces ordinary relations as well as $K$-relations and their associated basic properties. 

We use boldface letters to denote sets.
 For two sets 
 $\vect{X}$ and ${\vect{Y}}$, we write $\vect{X}\vect{\vect{Y}}$ to denote their union. 
If $A$ is an individual element, we sometimes write $A$ instead of $\{A\}$ to denote the singleton set consisting of $A$. 

\subsection{Relations}
Fix disjoint countably infinite sets $\att$ and $\val$ of variables and values. Each variable $A \in \att$
is associated with a subset of $\val$, called the \emph{domain} of $A$ and denoted $\Dom(A)$. 
Given a finite set of variables $\vect{X}$, an $\vect{X}$\emph{-tuple} is a mapping $t:\vect{X} \to \val$ such that $t(A) \in \Dom(A)$.
We write $\tup(\vect{X})$ for the set of all $\vect{X}$-tuples. Note that $\tup(\emptyset)$ is a singleton set 
consisting of the empty tuple. For $\vect{Y}\subseteq \vect{X}$, the \emph{projection} $t[\vect{Y}]$ of $t$ on $\vect{Y}$
is the unique $\vect{Y}$-tuple that agrees with $t$ on $\vect{X}$. In particular, $t[\emptyset]$
is always the empty tuple. 

A \emph{relation $R$} over $\vect{X}$ is a subset of $\tup(\vect{X})$.  
The variable set $\vect{X}$ is also called  the 
\emph{(relation) schema of $R$}. We sometimes write $R(\vect{X})$ instead of $R$ to
emphasize that $\vect{X}$ is the schema of $R$. 
For $\vect{Y}\subseteq \vect{X}$, 
the \emph{projection} of $R$ on $\vect{Y}$, written $R[\vect{Y}]$, is the set of all projections $t[\vect{Y}]$ where
$t\in R$. 
A \emph{database} $D$ is a finite collection of relations $\{R_1[\vect{X}_1], \dots ,R_n[\vect{X}_n]\}$.
Unless stated otherwise, we assume that each relation is finite.

\subsection{$K$-relations}
Fix a semiring $K$, and let $\vect{X}$ be a set of variables.
A \emph{$K$-relation} over $\vect{X}$ is a function $R : \tup(\vect{X}) \to K$. 
Again, the variable set $\vect{X}$ is called the 
\emph{(relation) schema of $R$}, and we can write $R(\vect{X})$ instead of $R$ to
emphasize that $\vect{X}$ is the schema of $R$.
If $K$ is the Boolean semiring $\mathbb{B}$,
 the tuple annotation $R(t)$ characterizes an ordinary relation, and
thus we will often in this paper identify $\mathbb{B}$-relations 
and relations. Note that a $K$-relation 
over $\emptyset$ associates the empty tuple with some value of $K$.
The \emph{support} $\Supp(R)$ of a $K$-relation $R$ over $\vect{X}$ is the set
$\{t\in \tup(\vect{X}) \mid R(t)\neq 0\}$
of tuples associated with a non-zero value. 
We often write $R'$ for the support of $R$.
The $K$-relation $R$ is called \emph{total} if for all $t \in \tup(\vect{X})$ it holds that $R(t)\neq 0$, 
i.e., if $\Supp(R) = \tup(\vect{X})$. It is called \emph{normal} if $\bigkplus_{t\in \tup(\vect{X})} R(t) =1$.
For $a \in K$, we write $aR$ for the $K$-relation over $\vect{X}$ defined by $(aR)(t) = aR(t)$.
For a $\vect{Y}$-tuple $t$, where $\vect{Y}\subseteq \vect{X}$, the \emph{marginal} of $R$ over $t$ is defined as
\begin{equation}\label{eq:marg}
  R(t)\coloneqq \bigkplus_{\substack{t'\in \tup(\vect{X})\\t'[\vect{Y}]=t}} R(t').
\end{equation}
We then write $R[\vect{Y}]$ for the relation over $\vect{Y}$, called the \emph{marginal} of $R$ on $\vect{Y}$,
that consists of the marginals of $R$ over all $\vect{Y}$-tuples. Note that the marginal $R[\emptyset]$ of $R$ on the empty set is a function that maps the empty tuple to $\sum_{t\in \tup(\vect{X})} R(t)$.
In particular, if $K$ is the Boolean 
semiring $\mathbb{B}$, the marginal of $R$ on $\vect{Y}$ is the projection of $R$ on $\vect{Y}$.
In this paper, we assume that each relation is finite and non-empty, and likewise each $K$-relation is assumed to have a finite and non-empty support.

$K$-relations instantiated in different ways lead to familiar notions. For instance, a database relation can be viewed
as $\mathbb B$-relation, and a probability distribution as a normal $\mathbb R_{\geq 0}$-relation. Alternatively, 
database relations can be transformed to $K$-relations by reinterpreting variables as tuple annotations. 

\begin{example}\label{ex:hotels}
Tab. \ref{tab:side_by_side_tables} collects data about room prizes in a hotel. The table can be viewed as a standard database relation.
Since Price is a function of Room, Date, and Persons, one can also interpret it as a $K$-relation Price(Room, Date, Persons) over some semiring $K$ containing
positive integers. In principle, other variables such as Room and Persons can also be turned into annotations.
\begin{table}[ht]
    \centering
    \begin{tabular}[t]{cccc}
        \textrm{Room} & \textrm{Date} & \textrm{Persons} & \textrm{Price} \\
        \midrule
        double & 2023-12-01 & 1 & 100 \\
        double & 2023-12-01 & 2 & 120 \\
        double & 2023-08-20 & 1 & 120 \\
        double & 2023-08-20 & 2 & 140 \\
        twin & 2023-08-20 & 1 & 110\\
        twin & 2023-08-20 & 2 & 120\\
        \bottomrule
    \end{tabular}
    \vspace{1mm}
    \caption{Price data for hotel rooms.}
    \label{tab:side_by_side_tables}
\end{table}
\end{example}

\subsection{Basic properties}
Prior to delving into the concept of conditional independence, we here list some basic properties regarding projections and supports of $K$-relations.
Lemmata \ref{lem:phokion} and \ref{lem:equiv} appear in \cite{Atserias20}, with the exception that there $K$ is always assumed to be positive.
Also the concept of a marginal in that paper is stated otherwise as in Eq. \eqref{eq:marg}, except that there $t'$ ranges over $R'$ instead of $\tup(X)$. Obviously the two versions lead to the same concept.
To account for these slight modifications, we include the proofs of these two lemmata 
Appendix \ref{sect:keq}.  

\begin{restatable}{lemma}{phokion}\label{lem:phokion}
Let $R(\vect{X})$ be a $K$-relation, and let $\vect{Z} \subseteq \vect{Y} \subseteq \vect{X}$. The following statements hold:
\begin{enumerate}
\item Assuming $K$ is $\kplus$-positive, for all $\vect{Y} \subseteq \vect{X}$ it holds that $R'[\vect{Y}] = R[\vect{Y}]'$.
\item For all $\vect{Z} \subseteq \vect{Y} \subseteq \vect{X}$ it holds that $R[\vect{Y}][\vect{Z}]=R[\vect{Z}]$.
\end{enumerate}
\end{restatable}

Two $K$-relations $R$ and $R'$ over a
variable set $\vect{V}$ are said to be \emph{equivalent (up to normalization)},
written $R \equiv R'$, if there are $a,b\in K\setminus\{0\}$ such that 
$aR=bR'$. 

\begin{restatable}{lemma}{simple}\label{lem:equiv}
  Let $K$ be a semiring, let $\vect{W}, \vect{V}$, $\vect{W}\subseteq \vect{V}$, be two variable sets, and let $R,R',R''$ be 
three $K$-relations over $\vect{V}$. Then,
\begin{enumerate}
  \item  $R\equiv R'$ implies $R[\vect{W}] \equiv R'[\vect{W}]$; and 
  \item if $K$ has no divisors of zero, $R \equiv R'$ and $R'\equiv R''$ implies $R \equiv R''$.
\end{enumerate}
\end{restatable}

\section{Conditional independence and decompositions}
Regardless of the context, what we call conditional independence tends to describe
essentially the same property. For a ``system'' consisting of three components $\vect{X},\vect{Y},\vect{Z}$, 
we might say that
$\vect{Y}$ is conditionally independent of $\vect{Z}$ given $\vect{X}$ if 
 $\vect{Y}$ does not reveal anything about $\vect{Z}$, once $\vect{X}$ has been fixed.
This usually entails that the ``system'' can be decomposed to its ``subsystems'' over $\vect{X},\vect{Y}$ and $\vect{X},\vect{Z}$ without loss of information.
 In this section we consider a general semantics for conditional independence over $K$-relations, and show that under certain assumptions, this definition matches the above intuition.
\begin{definition}[Conditional independence for $K$-relations \cite{BarlagHKPV23}]\label{def:ci}
Let $R$ be a $K$-relation over a variable set $\vect{V}$, and let $\vect{X},\vect{Y},\vect{Z}$ be  disjoint subsets of $\vect{V}$.
An expression of the form $\pci{\vect{X}}{\vect{Y}}{\vect{Z}}$ is called  a \emph{conditional independence} (CI).
We say that $R$ \emph{satisfies} $\pci{\vect{X}}{\vect{Y}}{\vect{Z}}$, denoted $R \models \pci{\vect{X}}{\vect{Y}}{\vect{Z}}$, if for all $\vect{V}$-tuples $t$,
\begin{equation}\label{eq:cidef}
R(t[\vect{X}\vect{Y}])  R(t[\vect{X}\vect{Z}]) = R(t[\vect{X}\vect{Y}\vect{Z}])  R(t[\vect{X}]).
\end{equation}
\end{definition}

Fix a relation schema $\vect{V}$ and three pairwise disjoint subsets $\vect{X}, \vect{Y},\vect{Z} \subseteq \vect{V}$.
 A \emph{saturated conditional independence} (SCI) is a
CI of the form $\pci{\vect{X}}{\vect{Y}}{\vect{Z}}$, where
$\vect{X}\vect{Y}\vect{Z}=\vect{V}$. Over $\mathbb{B}$-relations SCIs coincide with \emph{multivalued dependencies} 
 (MVDs), which are expressions of the form $\vect{X} \twoheadrightarrow \vect{Y}$, where $\vect{X} $ and $ \vect{Y}$ may overlap. 
  A $\vect{V}$-relation $R$ \emph{satisfies} $\vect{X} \twoheadrightarrow \vect{Y}$, written $R \models \vect{X} \twoheadrightarrow \vect{Y}$, if for all two tuples $t,t'\in R$ such that $t[\vect{X}]=t'[\vect{X}]$ there exists a third tuple $t''\in R$ such that $t''[\vect{X}\vect{Y}]=t'[\vect{X}\vect{Y}]$ and $t[\vect{V}\setminus \vect{X}\vect{Y}]=t'[\vect{V}\setminus \vect{X}\vect{Y}]$. 
 An \emph{embedded multivalued dependency} ({EMVD}) is an expression of the form $\vect{X} \twoheadrightarrow \vect{Y}\mid \vect{Z}$, where $\vect{X,Y,Z}$ may overlap.
 We say that $R$ \emph{satisfies} $\vect{X} \twoheadrightarrow \vect{Y}\mid \vect{Z}$, written $R \models \vect{X} \twoheadrightarrow \vect{Y}\mid \vect{Z}$, if the projection $R[\vect{X}\vect{Y}\vect{Z}]$ satisfies the MVD $\vect{X} \twoheadrightarrow \vect{Y}$.

\begin{example}
Returning to  Example \ref{ex:hotels}, we observe that the price function Price(Room, Date, Persons) exhibits certain types of dependencies between its arguments.
The room prices vary depending on the date and the room type. Additionally, adding a second person incurs a price increase by a flat rate which is independent of the date but  depends on the room type. 
 This kind of independence can be captured by viewing the price function as a $\mathbb{T}$-relation, in which case it satisfies the SCI $\pci{\textrm{Room}}{\textrm{Date}}{\textrm{Persons}}$. Suppose instead of a flat price increase, the addition of a second person incurs a 20$\%$ price increase for double rooms, and a 10$\%$ price increase for twin rooms. Then, interpreting Price(Room, Date, Persons) as a $\mathbb{R}_{\geq 0}$-relation, we again obtain Price $\models \pci{\textrm{Room}}{\textrm{Date}}{\textrm{Persons}}$. When Tab. \ref{tab:side_by_side_tables} is viewed as an ordinary relation, it satisfies the EMVD Room $ \twoheadrightarrow$ Date $\mid$ Persons, while failing to satisfy any MVD.
\end{example}


Several conditional independence notions from the literature can be recovered through $K$-relations. 
For instance, beside EMVDs, the following examples were considered in \cite{Studeny93} and can now be restated using the previous definition. 
\begin{itemize}
  \item For $K=\mathbb R_{\geq 0}$, the definition coincides with the concept of conditional independence in probability theory.
  \item For $K=\mathbb T$, the definition correponds to conditional 
  independence over natural conditional functions. A \emph{natural conditional function} 
  is a mapping $f \colon \tup(\vect{X})\to \mathbb N$, where $\min_{t\in \tup(\vect{X})} f(t)=0$. 
  The notion of conditional independence over such functions \cite{Studeny93}  coincides with Def. \ref{def:ci}
  over integral-valued, total, and normal $\mathbb T$-relations.
  Recall that for (min-plus) tropical semirings, addition is interpreted as minimum, and multiplication as the usual addition, meaning that its neutral element is $0$. 
  %
  \item For $K=\mathbb V$, the definition correponds to conditional independence over possibility functions. A \emph{possibility function}
  is a function $f \colon \tup(\vect{X})\to [0,1]$, where $\sum_{t\in \tup(\vect{X})} f(t)=1$. Such functions can be viewed as  
 normal $\mathbb V$-relations, where $\mathbb V$ is the Viterbi semiring, in which case their notion of conditional independence \cite{Studeny93} matches Def. \ref{def:ci}.
\end{itemize}



In order to connect conditional independence over $K$-relations to decompositions, we next consider the concept of a {join}. An arguably reasonable expectation  
is that whenever a $K$-relation $T(ABC)$ satisfies a CI $\pci{B}{A}{C}$, then one should be able to retrieve $T$ from its projections on $AB$ and $BC$ using the join. That is, $T$ should be equivalent to the join of $T[AB]$ and $T[BC]$ up to normalization.
In the relational context this is indeed the outcome once $\pci{B}{A}{C}$ is interpreted as the MVD $B \twoheadrightarrow A$, and
the {join} $R \bowtie S$ of two relations $R(\vect{X})$ and $S(\vect{Y})$ is given in the usual way, i.e., as 
the relation consisting of those $\vect{X}\vect{Y}$-tuples $t$ whose projections 
$t[\vect{X}]$ and $t[\vect{Y}]$ appear respectively in $R$ and $S$. 
 In the context of $K$-relations the join of $R(\vect{X})$ and $S(\vect{Y})$ is often defined via multiplication as the $K$-relation $R*S$ over $\vect{X}\vect{Y}$
where
 \begin{equation}\label{eq:multjoin}
 (R*S)(t) = R(t[\vect{X}])S(t[\vect{Y}])
 \end{equation}
 (see, e.g., \cite{GreenKT07}). Substituting $K=\mathbb{B}$ in this definition now yields the standard relational join. Similarly, letting 
$K=\mathbb{N}$ we arrive at
the bag join operation of SQL. 
However, as illustrated in the next example, this notion of a join falls short of our expectations.

\begin{example}
Continuing our running example, 
the two top tables in Fig. \ref{tab:projections} illustrate the projections of the $\mathbb{T}$-relation $\textrm{Price(Room, Date, Persons)}$ on $\{$Room, Date$\}$ and $\{$Room, Persons$\}$.
The table in the bottom row is the multiplicative join \eqref{eq:multjoin} of the two projections. Note that in the tropical semiring the aforementioned projections are formed as minima of prices, while addition plays the role of multiplication in the join operation. 
 We observe that the multiplicative join is not equivalent to the original price function. 
In particular, there is no uniform (tropical) scaling factor that returns us $\textrm{Price}$ from  $\textrm{Price}([\textrm{Room, Date}])*\textrm{Price}([\textrm{Room, Persons}])$.

\begin{figure}[ht]
    \centering
    \begin{tabular}{ccc}
    \begin{tabular}[t]{cc|c}
        \multicolumn{3}{l}{$\textrm{Price}[\textrm{Room, Date}]$} \\
        \toprule
        \textrm{Room} & \textrm{Date} & \emph{Price} \\
        \midrule
        double & 2023-12-01 & 100 \\
        double & 2023-08-20 & 120 \\
        twin & 2023-08-20 & 110 \\
        \bottomrule
    \end{tabular}
    &
    \quad
    &
        \begin{tabular}[t]{cc|c}
        \multicolumn{3}{l}{$\textrm{Price}[\textrm{Room, Persons}]$} \\
        \toprule
        \textrm{Room} & \textrm{Persons} & \emph{Price} \\
        \midrule
        double  & 1 & 100 \\
        double & 2 & 120 \\
        twin & 1 & 110\\
        twin & 2 & 120\\
        \bottomrule
    \end{tabular}
\\\\
        \begin{tabular}[t]{ccc|c}
        \multicolumn{4}{l}{$\textrm{Price}([\textrm{Room, Date}])*\textrm{Price}([\textrm{Room, Persons}])$} \\
        \toprule
        \textrm{Room} & \textrm{Date} & \textrm{Persons} & \emph{Price} \\
        \midrule
        double & 2023-12-01 & 1 & 200 \\
        double & 2023-12-01 & 2 & 220 \\
        double & 2023-08-20 & 1 & 220 \\
        double & 2023-08-20 & 2 & 240 \\
        twin & 2023-08-20 & 1 & 220 \\
        twin & 2023-08-20& 2 & 230 \\
        \bottomrule
    \end{tabular}
        &
    \quad
    &
        \begin{tabular}[t]{c|c}
        \multicolumn{2}{l}{$(\textrm{Price}[\textrm{Room}])^{-1}$} \\
        \toprule
        \textrm{Room} & \emph{Price} \\
        \midrule
        double &  -100 \\
        twin  & -110 \\
        \bottomrule
    \end{tabular}
    \end{tabular}
    \caption{Decomposition of the price function.}
    \label{tab:projections}
\end{figure}
\end{example}

Two $K$-relations $R(\vect{X})$ and $S(\vect{Y})$ are said to be \emph{consistent} if there exists a third relation $T(\vect{X}\vect{Y})$ such that $T[\vect{X}]\equiv R$ and $T[\vect{Y}]\equiv S$.
Atserias and Kolaitis \cite{Atserias20} demonstrate that the multiplicative join does not always witness the consistency of two $K$-relations, a fact that can be also seen from our running example. Consequently, they introduce a novel join operation which we will now incorporate into our approach. Intuitively this notion of a join is an adaptation of the factorization of a probability distribution obtained from conditional independence.  Suppose two random events $A$ and $C$ are independent given a third event $B$. The joint probability $P(A,B,C)$ can then be rewritten as $P(B)P(A\mid B)P(C\mid B)=P(A,B)P(B,C)/P(B)$. We may recognize that this equation is similar to the multiplicative join of two $K$-relations \emph{conditioned} on their common part. In our example this corresponds to multiplying the multiplicative join
$\textrm{Price}([\textrm{Room, Date}])*\textrm{Price}([\textrm{Room, Persons}])$ with the (tropical) multiplicative inverse of $\textrm{Price}([\textrm{Room}])$. We observe from Fig. \ref{tab:projections} that this sequence of operations yields the initial price function depicted in Tab. \ref{tab:side_by_side_tables} (even without re-scaling), in accordance with our expectations. 

We will now provide a precise definition of the join operation introduced in \cite{Atserias20}. This definition matches the above intuitive description with one exception: Semirings generally lack multiplicative inverses, and therefore the conditioning on the common part of two $K$-relations is defined indirectly. For a $K$-relation $R(\vect{X})$, a subset $\vect{Z}\subseteq \vect{X}$, and a $\vect{Z}$-tuple $u$, 
define 
\[
  c^*_{R,\vect{Z}} \coloneqq \bigktimes_{v\in R[\vect{Z}]'} R(v) \quad\text{   and   }\quad c_R(u) \coloneqq \bigktimes_{\substack{v\in R[\vect{Z}]'\\v \neq u}} R(v),
\]
with the convention that the empty product evaluates to $1$, the neutral element of multiplication in $K$. We isolate the following simple property which is applied frequently in the sequel.
\begin{proposition}\label{prop:notzero}
Suppose $K$ does not have divisors of zero. If $R(\vect{X})$ is a $K$-relation, $\vect{Z}\subseteq \vect{X}$, and $u$ is a $\vect{Z}$-tuple, then $c^*_{R,\vect{Z}} \neq 0$ and $c_R(u)\neq \emptyset$
\end{proposition}
  If $R(\vect{X})$ and $S(\vect{Y})$ are two $K$-relations, the \emph{join} $R\bowtie S$ of $R$ and $S$ is the $K$-relation over $\vect{X}\vect{Y}$ 
  defined by 
  \begin{equation}\label{eq:join}
    (R\bowtie S)(t)\coloneqq R(t[\vect{X}])S(t[\vect{Y}])c_S(t[\vect{X} \cap \vect{Y}]).
  \end{equation}
  If $K$ is a semifield (i.e., it has multiplicative inverses), we may rewrite the join as
  \[
      (R\bowtie S)(t)= \frac{c^*_{S,\vect{X}\cap \vect{Y}}R(t[\vect{X}])S(t[\vect{Y}])}{S(t[\vect{X} \cap \vect{Y}])}.
  \]

The definition of $R \bowtie S$ is not symmetric, and hence there may be occasions
where commutativity fails, i.e., $R \bowtie S \neq S \bowtie R$. However, whenever $R$ and $S$ agree on the 
 marginals on their shared variable set $\vect{X} \cap \vect{Y}$, commutativity holds by definition. In particular, Lemma \ref{lem:phokion} entails that
 the join of two projections 
$R[\vect{X}]$ and $R[\vect{Y}]$ of the same relation $R$ is commutative.

The join operation \eqref{eq:join} can also be described in terms of conditional independence and consistency. Suppose $K$ is a semifield, and suppose $\vect{X},\vect{Y},\vect{Z}$ are pairwise disjoint. Let $S(\vect{XY})$ and $T(\vect{XZ})$ be two normal
  $K$-relations that are consistent. Since equivalence entails identity for normal $K$-relations over semifields $K$, this is tantamount to finding a normal $K$-relation $R(\vect{XYZ})$ such that
\begin{equation}\label{eq:eqmax}
R[\vect{XY}]= S\textrm{ and }R[\vect{XZ}]= T.
\end{equation}
In particular, Lemma \ref{lem:phokion} and Eq. \eqref{eq:eqmax} yield $S[\vect{X}]= T[\vect{X}]$, whereby $c^*_{S,\vect{X}}=c^*_{T,\vect{X}}$. 
We may now observe that $R=1/c^*_{T,\vect{X}} (\join{S}{T})$ 
 is the unique 
 $K$-relation that satisfies \eqref{eq:eqmax} and the CI $\pci{\vect{X}}{\vect{Y}}{\vect{Z}}$. In the particular case where $K=\mathbb{R}_{\geq 0}$---in which case $S$ and $T$ are two consistent probability distributions---we also know that $R=1/c^*_{T,\vect{X}}(\join{S}{T})$ is the unique probability distribution that satisfies \eqref{eq:eqmax} and maximizes the entropy of $\vect{XYZ}$ (or alternatively, the conditional entropy of $\vect{YZ}$ given $\vect{X}$)
 \cite{Atserias20}.

Let $R$ be a $K$-relation over $\vect{X}\vect{Y}$. The 
\emph{decomposition of $R$ along $\vect X$ and $\vect Y$} consists of its projections 
 $R[\vect{X}]$ and $R[\vect{Y}]$  on $\vect X$ and $\vect Y$, respectively.
 Such a decomposition is called a \emph{lossless-join decomposition}
if $R[\vect{X}] \bowtie R[\vect{Y}] \equiv R$. This definition, which appears already in \cite{Atserias20}, generalizes the definition of a lossless-join decomposition 
in database relations.
It turns out, as we will next show, that if $K$ is positive and multiplicatively cancellative,
conditional independence holds on a $K$-relation if and only if 
the corresponding decomposition is a lossless-join one.

\begin{theorem}[Lossless-join decomposition]\label{thm:decomp}
  Let $K$ be a positive semiring, $\vect{X},\vect{Y},\vect{Z}$  pairwise disjoint sets of variables, and
   $R(\vect{X}\vect{Y}\vect{Z})$  a $K$-relation. If $R$ satisfies $\pci{\vect{X}}{\vect{Y}}{\vect{Z}}$,
  then the decomposition of $R$ along $\vect{X}\vect{Y}$ and $\vect{X}\vect{Z}$ is a lossless-join one. If $K$ is additionally multiplicatively cancellative, then the converse direction holds.
\end{theorem}
\begin{proof}
Assume $R$ satisfies $\pci{\vect{X}}{\vect{Y}}{\vect{Z}}$. 
  We need to show that $\join{R[\vect{X}\vect{Y}]}{R[\vect{X}\vect{Z}]} \equiv R$. 
  Let $t$ be an arbitrary tuple from $\tup(\vect{X}\vect{Y}\vect{Z})$. 
  By assumption and Lemma \ref{lem:phokion} we  obtain
  \begin{align*}
    (\join{R[\vect{X}\vect{Y}]}{R[\vect{X}\vect{Z}]}) (t) = &\, R(t[\vect{X}\vect{Y}])  R(t[\vect{X}\vect{Z}])  c_{R[\vect{X}\vect{Z}]}(t[\vect{X}])\\
    =&\, R(t[\vect{X}])  R(t)  \bigktimes_{\substack{v \in R[\vect{X}\vect{Z}][\vect{X}]'\\ v \neq t[\vect{X}]}} R[\vect{X}\vect{Z}](v)\\
    =&\, R(t[\vect{X}])  R(t) \bigktimes_{\substack{v \in R[\vect{X}]'\\ v \neq t[\vect{X}]}} R(v)\\
    =&\, c^*_{R,\vect{X}}  R(t),
  \end{align*}
  where $c^*_{R,\vect{X}} \neq 0$ by Proposition \ref{prop:notzero}. 
   This proves
  that $\join{R[\vect{X}\vect{Y}]}{R[\vect{X}\vect{Z}]} \equiv R$.
  
  For the converse direction, suppose $\join{R[\vect{X}\vect{Y}]}{R[\vect{X}\vect{Z}]} \equiv R$. Let $a,b \in K\setminus\{0\}$ be such that $aR = b(\join{R[\vect{X}\vect{Y}]}{R[\vect{X}\vect{Z}]})$. By Lemma \ref{prop:embedding}, we may assume without loss of generality that $K$ is a submodel of some positive semifield $F$. Hence $\join{R[\vect{X}\vect{Y}]}{R[\vect{X}\vect{Z}]} = cR$ for $c=ab^{-1} \in F$. We claim that $c=c^*_{R,\vect{X}}$. Since we assume a non-empty support for each $K$-relation, we may select a tuple $t$ from $R'$. By Lemma \ref{lem:phokion} we have $t[\vect{X}] \in R[\vect{X}]'$, i.e., $R(t[\vect{X}])\neq 0$. We can also deduce the following:
  \begin{align*}
  c  R(t[\vect{X}]) =&  \bigkplus_{\substack{t'\in \tup(\vect{X}\vect{Y}\vect{Z})\\t'[\vect{X}]=t[\vect{X}]}} cR(t')  =\bigkplus_{\substack{t'\in  \tup(\vect{X}\vect{Y}\vect{Z})\\t'[\vect{X}]=t[\vect{X}]}}(\join{R[\vect{X}\vect{Y}]}{R[\vect{X}\vect{Z}]})(t') \\
  =& \bigkplus_{\substack{t'\in  \tup(\vect{X}\vect{Y}\vect{Z})\\t'[\vect{X}]=t[\vect{X}]}}R(t'[\vect{X}\vect{Y}])  R(t'[\vect{X}\vect{Z}])  c_{R[\vect{X}\vect{Z}]}(t'[\vect{X}])\\
  =&\, c_{R[\vect{X}\vect{Z}]}(t'[\vect{X}])  \bigkplus_{\substack{t'\in  \tup(\vect{X}\vect{Y})\\t'[\vect{X}]=t[\vect{X}]}}R(t'[\vect{X}\vect{Y}])  \bigkplus_{\substack{t'\in  \tup(\vect{X}\vect{Z})\\t'[\vect{X}]=t[\vect{X}]}}R(t'[\vect{X}\vect{Z}]) \\
  =& \,  R(t[\vect{X}]) R(t[\vect{X}]) \bigktimes_{\substack{v \in R[\vect{X}]'\\ v \neq t[\vect{X}]}} R(v)= c^*_{R,\vect{X}} R(t[\vect{X}]).
  \end{align*}
  Multiplying (in $F$) by the inverse of $R(t[\vect{X}])$ then yields $c=c^*_{R,\vect{X}}$, proving our claim.
  
  Since $\join{R[\vect{X}\vect{Y}]}{R[\vect{X}\vect{Z}]} = c^*_{R,\vect{X}}R$, we may apply the sequence of equations from the previous case
  to obtain that for all $t \in \tup(\vect{X}\vect{Y}\vect{Z})$,
  \[
  R(t[\vect{X}\vect{Y}])  R(t[\vect{X}\vect{Z}])  c_{R[\vect{X}\vect{Z}]} = R(t[\vect{X}])  R(t)c_{R[\vect{X}\vect{Z}]}.
  \]
  Since $c_{R[\vect{X}\vect{Z}]}$ is non-zero by Proposition \ref{prop:notzero}, 
 it can be removed from both sides of 
  the equation by multiplicative cancellativity. We conclude that $R$ satisfies $\pci{\vect{X}}{\vect{Y}}{\vect{Z}}$.
\end{proof}

The preceding proof entails that over positive and multiplicatively cancellative semirings $K$, the satisfaction of $\pci{\vect{X}}{\vect{Y}}{\vect{Z}}$ by a $K$-relation $R(\vect{XYZ})$ holds if and only if $\join{R[\vect{X}\vect{Y}]}{R[\vect{X}\vect{Z}]} =  c^*_{R,\vect{X}} R$. If $K$ is additionally a semifield, then $R\models \pci{\vect{X}}{\vect{Y}}{\vect{Z}}$ exactly when we find two $K$-relations $S(\vect{X}\vect{Y})$ and $T(\vect{XZ})$ such that $R(t)=S(t[\vect{XY}])T(t[\vect{XZ}])$. 
 
Our running example demonstrates that semiring interpretations can give rise to lossless-join decompositions which are unattainable under the relational interpretation. 
\begin{example}
Consider again Tab. \ref{tab:side_by_side_tables} as a $\mathbb{T}$-relation Price(Room, Date, Persons). As we see from Fig. \ref{tab:projections},
this $\mathbb{T}$-relation decomposes along $\{$Room, Date$\}$ and $\{$Room, Persons$\}$. In particular, $\join{\textrm{Price[Room, Date]}}{\textrm{Price[Room, Persons]}}= c^*_{\textrm{Price}, \{\textrm{Room}\}}\textrm{Price}$ where $c^*_{\textrm{Price}, \{\textrm{Room}\}}=210$.
However, viewed as an ordinary relation $R$ over $\{$Price, Room, Date, Persons$\}$ 
this table is in sixth normal form, meaning that no decomposition along $X_1, \dots ,X_n$ is a lossless-join one, unless $X_i$ for some $i$ is the full variable set. Specifically, for any $X\subsetneq \{$Price, Room, Date, Persons$\}$, the projection of the tuple (double, 2023-08-20, 2, 120) on $X$ is in $R[X]$, even though the tuple itself does not belong to $R$.
\end{example}

Examples of positive semirings which are not multiplicatively cancellative seem somewhat artificial. Consider $K=(\mathbb{N}_{>0},\mathbb{{Z}}_2)\cup\{(0,0)\}$ with the semiring structure pointwise inherited from $\mathbb{N}$ and $\mathbb{{Z}}_2$. Note that $K$ is positive but violates multiplicative cancellativity, as 
$(1,1)\ktimes (1,0) = (1,0)\ktimes (1,0)$, while $(1,1) \neq (1,0)\neq (0,0)$. Using $K$ we can demonstrate that the assumption of multiplicative cancellativity cannot be dropped from the second statement of Lemma \ref{thm:decomp}.  We write $\pmi{\vect{X}}{\vect{Y}}$ for the \emph{marginal independence} between $\vect{X}$ and $\vect{Y}$, defined as the CI $\pci{\emptyset}{\vect{X}}{\vect{Y}}$.
Consider the $K$-relation $R$ from Fig. \ref{fig:noind}. This $K$-relation does not satisfy $\pmi{A}{B}$: Choosing $t(A,B)=(0,0)$ we observe $R(t[A]) \ktimes R(t[B])= (2,1) \ktimes (2,1) \neq (4,0)\ktimes (1,0)= R(\emptyset) \ktimes R(t[AB])$. On the other hand, we have $R\equiv \join{R[A]}{R[B]}$ because $aR=b(\join{R[A]}{R[B]})$, where $a=(4,0)\neq (0,0)\neq (1,0) =b$.

Similarly, the assumption of positivity is necessary for the first statement of Lemma \ref{thm:decomp}. Suppose $a,b\in K\setminus \{0\}$ are such that $a\kplus b=0$, and consider variables $X,Y$ with domain $\{0,1\}$. Then, the $K$-relation $R(XY)$ corresponding to the set $\{(0,0;a), (0,1;b), (1,0;b), (1,1;a)\}$ of triples $(t(X),t(Y);R(t))$ satisfies $\pmi{X}{Y}$, but the decomposition along $X$ and $Y$ is obviously not a lossless-join one. In fact, the definition of the marginal, Eq. \eqref{eq:marg}, may not even be useful if $K$ is not positive. For instance, a pure quantum state $\ket{\psi}_{XY}$ within a finite-dimensional composite Hilbert space $\mathcal{H}_{XY}$ can be conceived as a $\mathbb{C}$-relation $R(XY)$ over complex numbers $\mathbb{C}$. Its marginal with respect to $\mathcal{H}_X$ is however not obtained from Eq. \eqref{eq:marg}, but through a partial trace of the relevant density matrix. The marginal state may not even be a $\mathbb{C}$-relation anymore, because it can be mixed, i.e., a probability distribution over pure states.


\begin{figure}
\begin{center}
\begin{tabular}{ccccc}
\begin{tabular}[t]{cc|c}
\multicolumn{3}{l}{$R$}\\\toprule
$A$ & $B$  &$\#$\\ \midrule
$0$&$0$&$(1,0)$\\
$0$&$1$&$(1,1)$\\
$1$&$0$&$(1,1)$\\
$1$&$1$&$(1,0)$\\
\end{tabular}
&
\quad\quad\quad
&
\begin{tabular}[t]{c|c}
\multicolumn{2}{l}{$R[C], C \in \{A,B\}$}\\\toprule
$C$ &$\#$\\ \midrule
$0$&$(2,1)$\\
$1$&$(2,1)$\\
\end{tabular}
&
\quad\quad\quad
&
\begin{tabular}[t]{cc|c}
\multicolumn{3}{l}{$\join{R[A]}{R[B]}$}\\\toprule
$A$ & $B$ &$\#$\\ \midrule
$0$&$0$&$(4,1)$\\
$0$&$1$&$(4,1)$\\
$1$&$0$&$(4,1)$\\
$1$&$1$&$(4,1)$\\
\end{tabular}
\end{tabular}
\end{center}
\caption{Decomposition without independence. \label{fig:noind}}
\end{figure}

\section{Axiomatic properties}
The previous section identifies positivity and multiplicative cancellativity as the key semiring properties underlying 
the correspondence between conditional independence and lossless-join decompositions. 
 The main observation of the present section will be that the same 
key semiring properties guarantee soundness and completeness of central axiomatic properties associated with CIs.

\subsection{Semigraphoid axioms}
The \emph{semigraphoid axioms} \cite{Pearl88a} (the first five rules in Fig. \ref{fig:semi})
are a collection of fundamental conditional independence properties observed in various contexts, including database relations and probability 
distributions. The \emph{graphoid axioms} are obtained by extending the semigraphoid axioms with the interaction rule (the last rule in  Fig. \ref{fig:semi}). While
 not sound in general, the interaction rule is known to hold for probability distributions in which every probability is positive.
 We observe next that these results extend to $K$-relations whenever
$K$ is positive and multiplicatively cancellative; the interaction rule, in particular, is sound over total $K$-relations.


\begin{figure}
  \centering
  \begin{tikzpicture}[every node/.style={outer sep=1pt}]
  \def\m{1.4em}
    \node[draw,minimum width=0.5\textwidth, rounded corners, minimum height=3.2cm] (box1) {
    \hspace{.3cm}
      \begin{minipage}[t][2.6cm]{.9\textwidth}
        \begin{itemize}
          \item[(S1)] Triviality: $\pci{\vect{X}}{\vect{Y}}{\emptyset}$. 
          \item[(S2)] Symmetry: $\pci{\vect{X}}{\vect{Y}}{\vect{Z}}$, then $\pci{\vect{X}}{\vect{Z}}{\vect{Y}}$. 
          \item[(S3)]  Decomposition: $\pci{\vect{X}}{\vect{Y}}{\vect{Z}\vect{W}}$, then $\pci{\vect{X}}{\vect{Y}}{\vect{Z}}$. 
          \item[(S4)]  Weak union: $\pci{\vect{X}}{\vect{Y}}{\vect{Z}\vect{W}}$, then $\pci{\vect{X}\vect{Z}}{\vect{Y}}{\vect{W}}$. 
          \item[(S5)]  Contraction: $\pci{\vect{X}}{\vect{Y}}{\vect{Z}}$ and $\pci{\vect{X}\vect{Z}}{\vect{Y}}{\vect{W}}$, then $\pci{\vect{X}}{\vect{Y}}{\vect{Z}\vect{W}}$. 
          \item[(G)]  Interaction: $\pci{\vect{X}\vect{W}}{\vect{Y}}{\vect{Z}}$ and $\pci{\vect{X}\vect{Z}}{\vect{Y}}{\vect{W}}$, then $\pci{\vect{X}}{\vect{Y}}{\vect{Z}\vect{W}}$. 
          \end{itemize}
  \end{minipage}
    };
  \end{tikzpicture}
\caption{Semigraphoid axioms (S1-S5) and graphoid axioms (S1-S5,G).  \label{fig:semi}}
\end{figure}

%

To offer context for Theorem \ref{thm:sound:pci} which is proven in Appendix \ref{sect:graphoid}, recall from information theory  the concept of \emph{conditional mutual information}, which can be defined over sets of random variables $\vect{U,V,W}$ as
\begin{equation}\label{eq:condmutual}
I_h(\vect{V};\vect{W}\mid \vect{U}) \coloneqq h(\vect{UV})+h(\vect{UW})-h(\vect{U})-h(\vect{UVW}),
\end{equation} 
Assuming $h$ is the Shannon entropy (discussed in Sect. 7), Eq. \eqref{eq:condmutual}
 is zero if and only if the CI  $\pci{\vect{U}}{\vect{V}}{\vect{W}}$ holds in the underlying probability distribution.
Now, consider the \emph{chain rule}
\begin{equation*}
  I_h(\vect{Y};\vect{Z} \mid \vect{X}) + I_h(\vect{Y};\vect{W}\mid \vect{X}\vect{Z}) = I_h(\vect{Y};\vect{Z}\vect{W}\mid \vect{X})
\end{equation*}
of conditional mutual information. 
Since conditional mutual information is non-negative, the chain rule readily entails Decomposition, Weak union, and Contraction 
for probability distributions. In the semiring setting we cannot deduce these rules analogously, as there seems to be no general  measure to capture 
conditional independence over $K$-relations. We can however use the measure-theoretic interpretation of conditional independence as a guide toward a proof. 
 Consider, for instance, the contraction rule, which can be restated in the information context as follows:
if 
\begin{align}
  h(\vect{X}\vect{Y}) + h(\vect{X}\vect{Z}) &= h(\vect{X})+ h(\vect{X}\vect{Y}\vect{Z}) \text{ and}\label{eq:ant1}\\
  h(\vect{X}\vect{Y}\vect{Z}) + h(\vect{X}\vect{Z}\vect{W}) &= h(\vect{X}\vect{Z}) + h(\vect{X}\vect{Y}\vect{Z}\vect{W}),\label{eq:ant2}
\end{align}
then 
\begin{equation}\label{eq:cons}
  h(\vect{X}\vect{Y}) + h(\vect{X}\vect{Z}\vect{W}) = h(\vect{X}) + h(\vect{X}\vect{Y}\vect{Z}\vect{W}).
\end{equation}
In particular, Eq. \eqref{eq:cons} is a consequence of subtracting $h(\vect{X}\vect{Z})+h(\vect{X}\vect{Y}\vect{Z})$ from the combination of Eqs. \eqref{eq:ant1} and \eqref{eq:ant2}.
The soundness proof for $K$-relations has now the same general structure.
Instantiations of Eq. \eqref{eq:cidef} for the CIs appearing in the contraction rule are structurally similar to Eqs.  \eqref{eq:ant1},  \eqref{eq:ant2}, and \eqref{eq:cons}, with addition between entropies being replaced by multiplication within $K$. 
Instead of subtraction, one now applies multiplicative cancellativity to remove all superfluous terms from the combination of two equations.
Additionally, one has to deal with those cases where the terms to be eliminated are zero, and multiplicative cancellativity cannot be applied.

Given a set $\Sigma \cup\{\tau\}$ of formulae, we say that $\Sigma$ \emph{implies} $\tau$ over relations (resp. $K$-relations), denoted $\Sigma \models \tau$ (resp. $\Sigma \models_K \tau$), if every relation (resp. $K$-relation) satisfying $\Sigma$ satisfies $\tau$.
Given a semiring $K$ and a list of formulae  $\sigma_1, \dots ,\sigma_n,\tau$,
we say that an inference rule of the form 
\begin{equation}\label{eq:rule}
  \text{ if $\sigma_1, \dots ,\sigma_n$, then $\tau$}
\end{equation}
is \emph{sound} for relations (resp. $K$-relations) if  $\Sigma$ implies $\tau$ over relations (resp. $K$-relations).
\begin{restatable}{theorem}{sound}\label{thm:sound:pci}
Triviality, Symmetry, and Decomposition are sound for $K$-relations. Weak union and Contraction are sound for $K$-relations where $K$ is positive and multiplicatively cancellative. 
Interaction is sound for total $K$-relations where $K$ is positive and multiplicatively cancellative.
\end{restatable}

Having considered the graphoid axioms for $K$-relations, we next consider the interaction between conditional independence and functional dependence.

\subsection{Functional dependencies}
Given two sets of variables $\vect{X}$ and $\vect{Y}$, the expression $\vect{X} \to \vect{Y}$ is called a \emph{functional dependency} (FD). A relation $R$ \emph{satisfies} $\vect{X} \to \vect{Y}$, denoted $R\models  \vect{X} \to \vect{Y}$, if for all $t,t'\in R$, $t[\vect{X}]=t'[\vect{X}]$ implies $t[\vect{Y}]=t'[\vect{Y}]$. We extend this definition to $K$-relations $R$ by stipulating that $R$ \emph{satisfies} an FD $\sigma$ whenever its support $R'$ satisfies $\sigma$. 

The Armstrong axioms for FDs  \cite{armstrong74} comprise the first three rules in Fig. \ref{fig:arm}. These rules are sound and complete for database relations, and hence, by definition, for $K$-relations over any  $K$. The last two rules are two combination rules for MVDs and FDs \cite{BeeriFH77} rewritten in different syntax. 
  To extend these rules $K$-relations, we again need positivity and multiplicative cancellativity. The following proposition is proven in 
  Appendix \ref{sect:combined}.

\begin{figure}
  \centering
  \begin{tikzpicture}[every node/.style={outer sep=1pt}]
  \def\m{1.4em}
    \node[draw,minimum width=0.3\textwidth, minimum height=2.8cm, rounded corners] (box1) {
    \hspace{1.2cm}
      \begin{minipage}[t][2.2cm]{.8\textwidth}
\begin{itemize}
\item[(FD1)] Triviality: if $\vect{Y}\subseteq \vect{X}$, then $\vect{X}\to \vect{Y}$.
\item[(FD2)] Augmentation: if $\vect{X}\to \vect{Y}$ and $\vect{XZ}\to \vect{YZ}$.
\item[(FD3)] Transitivity:  if $\vect{X}\to \vect{Y}$ and $\vect{Y}\to \vect{Z}$, then $\vect{X}\to \vect{Z}$.
\item[(FD-CI1)] CI introduction: if $ \vect{X}\to \vect{Y}$, then $\pci{\vect{X}}{\vect{Y}}{\vect{Z}}$.
\item[(FD-CI2)] FD contraction: if $\pci{\vect{X}}{\vect{Y}}{\vect{Z}}$ and $\vect{XY}\to \vect{Z}$, then $\vect{X}\to \vect{Z}$.
\end{itemize}
  \end{minipage}
    };
  \end{tikzpicture}
\caption{Armstrong's axioms (FD1-FD3) and combination rules (FD-CI1,FD-CI2) \label{fig:arm}}
\end{figure}

\begin{restatable}{proposition}{cifd}\label{prop:fdmdv}
CI introduction is sound for all $K$-relations, where $K$ is $\kplus$-positive. FD contraction is sound for all $K$-relations, where $K$ is positive and multiplicatively cancellative. 
\end{restatable}

Soundness of CI introduction means that, for positive $K$, a functional dependency on a $K$-relation leads to a lossless-join decomposition. The next proposition, stating this fact, was proven originally in \cite{Atserias20}.
 Alternatively, we now see that the proposition follows directly by Theorem \ref{thm:decomp} and Proposition \ref{prop:fdmdv}.
\begin{proposition}[\cite{Atserias20}]
Let $K$ be a positive semiring, $\vect{X},\vect{Y},\vect{Z}$ be pairwise disjoint sets of variables, and
   $R(\vect{X}\vect{Y}\vect{Z})$ be a $K$-relation. If $R$ satisfies $\vect{X} \to \vect{Y}$,
  then the decomposition of $R$ along $\vect{X}\vect{Y}$ and $\vect{X}\vect{Z}$ is a lossless-join one.
\end{proposition}

We have now examined fundamental inference rules for CIs and FDs that have their origins in database theory and probability theory. The combination of these rules however is not---and cannot be---complete in either context. Specifically, over both finite relations and finite distributions, the implication problems for EMVDs/CIs are not even r.e.
 since the problems are known to be undecidable \cite{herrmann95,kuhne22,Li23} and co-r.e. \cite{KhamisK0S20}. 
 In the next section we restrict attention to saturated CIs which are known to exhibit favorable algorithmic and axiomatic properties. 

\subsection{Saturated conditional independence and functional dependence}
We next show that SCI+FD enjoys a complete axiomatization that is shared by all positive and multiplicatively cancellative semirings $K$. This result readily entails that logical implication within the class SCI+FD does not depend on the chosen semiring $K$, provided that it has the fundamental properties mentioned above.

Given a set $\Sigma \cup\{\tau\}$ of dependencies, we say that $\Sigma$ \emph{implies} $\tau$ over relations (resp. $K$-relations), denoted $\Sigma \models \tau$ (resp. $\Sigma \models_K \tau$), if every relation (resp. $K$-relation) satisfying $\Sigma$ satisfies $\tau$. Let $\sigma \mapsto \sigma^*$  associate an SCI/CI with its corresponding MVD/EMVD. 
Extend this mapping to be the identity on FDs, and extend it to sets in the natural way: $\Sigma^* = \{\sigma^*\mid \sigma\in \Sigma\}$.
\begin{theorem}\label{thm:kenig}
Let $K$ be a positive and multiplicatively cancellative semiring.
Let $\Sigma\cup\{\tau\}$ be a set of SCIs and FDs. 
 The following are equivalent:
\begin{enumerate}
\item $\tau$ can be derived from $\Sigma$ using (S1-S5), (FD1-FD3), and (FD-CI1,FD-CI2). 
\item $\Sigma$ implies $\tau$ over $K$-relations.
\item $\Sigma^*$ implies $\tau^*$ over  relations consisting of two tuples.
\item $\Sigma^*$ implies $\tau^*$ over  relations.
\end{enumerate}
\end{theorem}
\begin{proof}
(1) $\Rightarrow$ (2). This direction is immediate due to Theorem \ref{thm:sound:pci}, Proposition \ref{prop:fdmdv}, and soundness of the Armstrong axioms for ordinary relations. (2) $\Rightarrow$ (3). Any two-tuple relation $R=\{t,t'\}$ can be transformed to a $K$-relation $S$ such that the support $S'$ is $R$, and $S(t)=S'(t)=1$. It is straightforward to verify that $R$ satisfies $\sigma^*$ if and only if $S$ satisfies $\sigma$, for all CIs and FDs $\sigma$. From this, the direction follows.
(3) $\Rightarrow$ (4). This direction has been proven in \cite{SagivDPF81}. (4) $\Rightarrow$ (1). This direction follows from the fact that the system (S1-S5), (FD1-FD3), (FD-CI1,FD-CI2) mirrors the complete axiomatization of MVDs and FDs. We give an explicit proof in 
Appendix \ref{sect:axiomssci}.
\end{proof}
\begin{corollary}\label{cor:korro}
Let $K,K'$ be positive and multiplicatively cancellative semirings, and let $\Sigma\cup\{\tau\}$ be a set of SCIs and FDs. Then, $\Sigma$ implies $\tau$ over $K$-relations if and only if $\Sigma$ implies $\tau$ over $K'$-relations.
\end{corollary}

Theorem \ref{thm:decomp} and Corollary \ref{cor:korro} demonstrate that the decomposition properties arising from multivalued and functional dependencies hold invariably for all $K$-relations, given $K$ is multiplicatively cancellative and positive.
Standard database normalization methods thus extend to diverse contexts and may sometimes coincide
with existing methods.  The following example shows that relational normalizations can sometimes match the factorizations of probability distributions arising from Bayesian networks.

A relational database schema is a set of relation schemata, each associated with a set of constraints.
It is
in 
\emph{fourth normal form} (4NF) if for any of its MVD constraints $\vect{X} \twoheadrightarrow \vect{Y}$, $\vect{X}$ is a superset of a key. 
A \emph{Bayesian network} is a directed acyclic graph in which the nodes represent random variables and the directed edges probabilistic dependencies between variables. 
Each node is thus directly influenced by its parents in the graph. Conversely, each node indirectly influences its descendants by transitivity. The \emph{local Markov property} states that once the parent nodes are known, the state of the current node does not reveal any additional information about the states of its non-descendants, i.e.,
each node is conditionally independent of its non-descendants given its parents.

\begin{example}
Consider the Bayesian network in Fig. \ref{fig:tikz}. The chain rule of probability distributions and the local Markov property implies that the joint distribution
$P(A,B,C,D,E)$ has a factorization $P(A)P(B\mid A)P(C\mid A)P(D\mid BC)P(E\mid D)$. 

The local Markov property produces three non-trivial CIs up to symmetry, rewritten as the following EMVDs
 $A \twoheadrightarrow B | C$, $BC \twoheadrightarrow A|D$, $D \twoheadrightarrow ABC | E$. 
 Suppose our goal is to transform the unirelational database schema $\{ABCDE\}$ into 4NF, assuming absence of key constraints. Since the last EMVD is also an MVD, we first decompose $ABCDE$ along $ABCD$ and $DE$. Since the second EMVD is an MVD on $ABCD$, we continue by splitting $ABCD$ into $ABC$ and $BCD$. To remove the last remaining MVD, we decompose $ABC$ along $AB$ and $AC$. The final schema $\{AB,AC,BCD,DE\}$ is free of MVDs, and thus in 4NF. Furthermore, the decomposition of $P$ (as a $\mathbb{R}_{\geq 0}$-relation) along $\{AB,AC,BCD,DE\}$ reproduces the aforementioned factorization of $P$ into conditional probabilities.


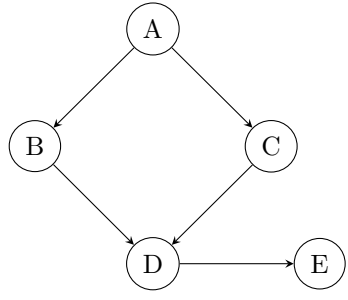
\begin{figure}[h!]
\centering
\begin{tikzpicture}[>=stealth,->, every node/.style={circle, draw}]
  \node (A) at (0,0) {A};
  \node[below left=1.5cm of A] (B) {B};
  \node[below right=1.5cm of A] (D) {C};
  \node[below left=1.5cm of D] (E) {D};
  \node[right=1.5cm of E] (C) {E};
  
  \draw[->] (A) -- (B);
  \draw[->] (A) -- (D);
  \draw[->] (B) -- (E);
  \draw[->] (E) -- (C);
  \draw[->] (D) -- (E);
\end{tikzpicture}
\caption{A simple Bayesian network. \label{fig:tikz}}
\end{figure}

\end{example}

 \section{Comparison of implication}
As mentioned previously, implication for non-saturated CIs depends heavily on the underlying semantics. In this section, we examine the connections between different conditional independence semantics in relation to the semiring properties they rely on. Using model-theoretic arguments, we first show that $\Sigma\models_{\mathbb R_{\geq 0}} \tau$ implies $\Sigma\models_{K} \tau$, whenever $K$ is cancellative and equipped with a natural total order.


Consider a CI of the form $\tau=\pci{\vect{X}}{\vect{Y}}{\vect{Z}}$, where $\vect{X},\vect{Y},\vect{Z}$ are disjoint subsets of a schema $\vect{V}$. Suppose the domain of each variable in $\vect{V}$ is finite.
For each $\vect{V}$-tuple $t$, introduce
a variable $x_t$. Denote by $\vec{x}_{\vect{V}}$ a sequence listing all variables $x_t$, $t\in \tup(\vect{V})$. 
We associate $\tau$ and $\vect{V}$ with a quantifier-free first-order arithmetic formula
\[
  \phi_{\tau,\vect{V}} \coloneqq \bigwedge_{\substack{t_{\vect{X}}\in \tup(\vect{X})\\t_{\vect{Y}}\in \tup(\vect{Y})\\t_{\vect{Z}}\in \tup(\vect{Z})}}
  \bigkplus_{\substack{t\in \tup(\vect{XYZ})\\t[\vect{X}]=t_{\vect{X}}}} x_t  \bigkplus_{\substack{t\in \tup(\vect{XYZ})\\t[\vect{X}]=t_{\vect{X}}\\t[\vect{Y}]=t_{\vect{Y}}\\t[\vect{Z}]=t_{\vect{Z}}}} x_t =
    \bigkplus_{\substack{t\in \tup(\vect{XYZ})\\t[\vect{X}]=t_{\vect{X}}\\t[\vect{Y}]=t_{\vect{Y}}}} x_t   \bigkplus_{\substack{t\in \tup(\vect{XYZ})\\t[\vect{X}]=t_{\vect{X}}\\t[\vect{Z}]=t_{\vect{Z}}}} x_t.
\]

A formula $\phi$ is said to be \emph{universal} if it is of the form $\forall x_1 \ldots \forall x_n \theta$,
where $\theta$ is quantifier-free. All universal first-order properties of a model are preserved
 for
its submodels (and any of their isomorphic copies)
\cite{DBLP:books/daglib/0067423}.
\begin{proposition}\label{lem:embedding}
  Let $\calA$ and $\calB$ be models over a vocabulary $\tau$, and let $\phi$ be a 
  universal first-order sentence over $\tau$. If $\calA$ embeds in $\calB$, then 
  $\calA\models \phi$ implies $\calB \models \phi$.
\end{proposition}

The following theorem lists some basic properties of real-closed fields \cite{ Bochnak98,DBLP:books/daglib/0067423}. 
A field $F$ is called \emph{real} if it can be associated with an ordering $\leq$
such that $(F,\leq)$ becomes an ordered field. A \emph{field} $F'$ is an
\emph{extension} of a field $F$ if $F\subseteq F'$, and the field operations of $F$ are those inherited
from $F'$. The extension is \emph{proper} if $F$ is a strict subset of $F'$, 
and \emph{algebraic} if every element in $F'$ is a root of a non-zero polynomial
with coefficients in $F$. A real field with no proper real
algebraic extension is called \emph{real closed}. For instance, the field of real numbers
is real closed, whereas the field of rational numbers is real but
not real closed. 
On the other hand, no finite or algebraically closed field is real.
A real closed field has a unique ordering, which is definable by $a \leq b \logeq \exists c (a \kplus c^2 = b)$. 
An algebraic extension $F'$ of an ordered field $(F,\leq)$ is called a \emph{real closure} of $F$ if $F'$ is real closed, and its unique ordering extends that of $F$ (i.e., the ordering is preserved under the inclusion map $F \xhookrightarrow{} F'$).
Two models $\calA$ and $\calB$ are \emph{elementarily equivalent}, written $\calA \equiv \calB$, if they satisfy the same first-order sentences.
 \begin{theorem}\label{thm:basic}
 \text{}\\
 \vspace{-5mm}
\begin{itemize}
\item Any totally ordered field $(F,\leq)$ has a real closure $F'$.
 \item If $(F, \leq)$ is a totally ordered field, and $F_0$ and $F_1$ are its real closures uniquely ordered by $\leq_0$ and $\leq_1$, there is an isomorphism between $(F_0, \leq_0)$ and $(F_1, \leq_1)$ which is identity on $F$.
 \item Any two real-closed fields $F_0$ and $F_1$ are elementarily equivalent.
 \end{itemize}
\end{theorem}

We can now prove the property that $\mathbb R_{\geq 0}$-implication entails $K$-implication, for any semiring $K$ embedded in a cancellative and naturally totally ordered one.
\begin{theorem}\label{thm:rpostoK}
  Let $\Sigma \cup\{\tau\}$ be a finite set of CIs, and suppose $K$ embeds in a 
  naturally totally ordered 
  cancellative semiring. Then, $\Sigma \models_{\Rpos} \tau$ implies $\Sigma \models_K \tau$.
\end{theorem}
\begin{proof}
  By Lemma \ref{lem:semi-cons}, Theorem \ref{thm:basic}, and transitivity of the embedding relation, $K$ embeds in a real-closed field $F$. 
  Let $R(\vect{V})$ be a $K$-relation, where $\vect{V}$ is a set of variables that includes each variable appearing in $\Sigma\cup\{\tau\}$. We need to show that $R \models \Sigma$ implies $R \models \tau$. Since satisfaction of a CI  by $R$ does not depend on tuple values that do not appear in $R$, we may 
  without loss of generality assume that the domain of each variable in $\vect{V}$ is finite.
  Then, assuming $\Sigma =\{\sigma_1, \dots ,\sigma_n\}$, we may consider the universal first-order sentence
\[
  \psi_{\Sigma, \tau, \vect{V}} \coloneqq \forall \vec{x}_{\vect{V}}( \vec{x}_{\vect{V}} \geq \vec{0} \land \phi_{\sigma_1,\vect{V}} \land \ldots \land \phi_{\sigma_n,\vect{V}} \to \phi_{\tau,\vect{V}}),
\]
where, given a sequence $\vec{x} = (x_1, \dots ,x_l)$, we write  $\vec{x} \geq \vec{0}$ as a shorthand for $x_1 \geq 0 \land \ldots \land x_l \geq 0$.
Since $\Sigma \models_{\Rpos} \tau$ by hypothesis, $\psi_{\Sigma, \tau, \vect{V}}$ must be true
for the field of reals $\mathbb{R}$ (equipped with its unique ordering). Since $F$ and $\mathbb{R}$ are real-closed,
they share the same first-order properties; in particular,
$F$ also satisfies $\psi_{\Sigma, \tau, \vect{V}}$. By Proposition \ref{lem:embedding},
$K$ likewise satisfies $\psi_{\Sigma, \tau, \vect{V}}$, and hence $R\models \Sigma$ implies $R \models \tau$.
We conclude that $\Sigma \models_K \tau$.
\end{proof}

The preceding theorem readily entails that implication over $\mathbb{R}_{\geq 0}$ entails implication over $\mathbb{N}_{\geq 0}$ and $\mathbb{Q}_{\geq 0}$.
 Another example is $K=\mathbb{N} \times \mathbb{N}$ with pointwise addition and multiplication, and neutral elements $(0,0)$ and $(1,1)$. This semiring is not naturally totally ordered, because it contains incomparable elements,
such as $(0,1)$ and $(1,0)$. However, it can be extended to $K\cup (\mathbb{{Z}} \times \mathbb{N}_{>0})$ which is naturally totally ordered and cancellative. For another example,
consider the  semiring $K=\mathbb{N}[X]$ of polynomials in ${X}$ with coefficients from natural numbers. Since $K$ contains incomparable elements, such as $X+2$ and $2X+1$, its natural order is not total. The extension of $K$ with those polynomials of $\mathbb{{Z}}[{X}]$ in which the leading coefficient is positive produces a cancellative 
semiring whose natural order is total.

Let us then turn attention to the Boolean semiring $\mathbb B$ and its connections with other semirings. 
First we note that although both $\mathbb{R}_{\geq 0}$ and $\mathbb B$ are naturally totally ordered, only the first one is additively cancellative. In light of Theorem \ref{thm:rpostoK}, this may help explain why there is no implication from $\Sigma\models_{\mathbb R_{\geq 0}} \tau$ to $\Sigma\models_{\mathbb B} \tau$. Another difference is that only the Boolean semiring is associated with an idempotent addition; an operation $*$ on $K$ is said to be \emph{idempotent} if $a * a =a$ for all $a\in K$. 
 We observe that $\mathbb B$-implication entails $K$-implication, whenever $K$ has an idempotent addition.
\begin{proposition}\label{prop:idempotenttoB}
  Let $K$ be a semiring associated with an idempotent addition. 
   Let $\Sigma \cup \{\tau\}$ be a 
  set of CIs. Then, $\Sigma \models_K \tau $ implies $\Sigma \models_{\mathbb{B}} \tau$.
\end{proposition}
\begin{proof}
  Recall that we consider only non-trivial semirings $K$, where $0\neq 1$. Thus, any $\mathbb{B}$-relation $R$ can be readily interpreted as 
  a $K$-relation $R'$. The idempotence of addition guarantees that $R \models \sigma$ if and only if $R'\models \sigma$, for any 
  CI $\sigma$.   The statement of the lemma then follows.
\end{proof}

We leave it as an open question whether the statements of Theorem \ref{thm:rpostoK} and Proposition \ref{prop:idempotenttoB} hold also in the converse directions.

\begin{example}\label{ex:cycle}
Consider the following \emph{cycle rule} for conditional independence:
\begin{equation*}
\text{ if $\pci{{A}_2}{{A}_1}{{B}}$, $\pci{{A}_3}{{A}_2}{{B}}$, $\ldots$ ,$\pci{{A}_{n}}{{A}_{n-1}}{{B}}$, $\pci{{A}_{1}}{{A}_n}{{B}}$, then $\pci{{A}_1}{{A}_2}{{B}}$.}
\end{equation*}
 Under various semantics, this rule is both sound and irreducible to smaller rules, i.e., rules with less than $n$ many antecedents. 
Proving this has been the common strategy to establish that conditional independence lacks finite axiomatic characterization
(see \cite{parker:1980,studeny:1993} for the database and probability cases, respectively). As the cycle rule is sound over $\mathbb R_{\geq 0}$-relations, Theorem \ref{thm:rpostoK} implies in particular that this rule remains sound under any cancellative and totally ordered semiring $K$. However, it is less straightforward to determine whether the cycle rule is sound for positive and additively idempotent semirings other than $\mathbb B$. One example of such a semiring is the tropical semiring $\mathbb T$. 
As part of a non-axiomatizability proof, Studen\'{y} \cite{Studeny95} has shown that the cycle rule is sound for the natural conditional functions, which can be viewed as a special subclass of total $\mathbb T$-relations.
The proof however relies on the existence of multiplicative inverses (in the sense of $\mathbb T$), and thus does not extend to arbitrary $\mathbb T$-relations. 
 Whether the cycle rule is sound under the tropical semiring, is left as an open question.
\end{example}

\section{Chase and Copy Lemma}\label{sect:copylemma}
Having considered the relationships between different CI semantics, we next compare the methods used to test validity and logical consequence in this context.
In databases, the \emph{chase} constitutes the key algorithm---or family of algorithms---to test logical implication between database dependencies. In information theory, the \emph{copy lemma}, presented next, is the main tool to find non-trivial inequalities between entropies; recall that conditional independence over random variables can in particular be expressed as an inequality of this form.

\subsection{Copy Lemma}
We begin by reviewing some central concepts from information theory.
Let $X$ be a random variable associated with a finite domain $D=\Dom(X)$ and a probability distribution $p:D \to [0,1]$, where $\sum_{a\in D} p(a)=1$.
The \emph{entropy} of $X$  is defined as
\begin{equation}\label{eq:entropy}
H(X) \coloneqq -\sum_{x\in D} p(x) \log p(x).
\end{equation}
Entropy is non-negative and does not exceed the logarithm of the domain size: $0\leq H(X) \leq \log |D|$. In particular, $H(X)=0$ if and only if $X$ is constant (i.e., $p(a)=1$ for some $a \in D$), and $H(X)=\log |D|$ if and only if $X$ is uniformly distributed (i.e., $p(a) = 1/|D|$ for all $a \in D$). 

 %
 

 Consider now a normal  $\mathbb{R}_{\geq 0}$-relation $R$ over a set $\vect{X}=\{X_i\}_{i=1}^n$ of variables with finite domains; in other words, $R$ is a joint probability distribution of random variables with a finite domain. Each subset of $\vect{Y}\subseteq \vect{X}$ can be viewed as 
 a random variable with domain $D=R[\vect{Y}]'$ and probability distribution $R[\vect{Y}]$, which is the marginal of $R$ on $\vect{Y}$.
 The $\mathbb{R}_{\geq 0}$-relation $R$ then defines an \emph{entropic function} 
 ${h}=(h_\alpha)_{
 \alpha \subseteq [n]}$,%
 where $h_\alpha\coloneqq H(\vect{X}_\alpha)$. 
  The \emph{entropic region} $\Gamma^*_n\subseteq \mathbb{R}^{2^{n}}$ consists of all entropic functions over $n$, and 
  the \emph{almost entropic region} $\overline{\Gamma^*_n}$ is defined as the topological closure of $\Gamma^*_n$.
  Writing  $h(\vect{X}_\alpha)$ for $h_\alpha$, the entropic functions satisfy the polymatroidal axioms depicted in Fig. \ref{fig:polym}.
    \begin{figure}[ht]\label{fig:axioms}
  \centering
  \begin{tikzpicture}[every node/.style={outer sep=0pt}]
  \def\m{1.4em}
    \node[draw, rounded corners, minimum width=.9\textwidth] (box1) {
      \begin{minipage}[t][1.8cm]{0.85\textwidth}
      \hspace{0cm}
        \begin{itemize}
  \item $h(\emptyset) = 0$  
  \item $h(\vect{X}\cup \vect{Y}) - h(\vect{X})\geq 0$ (monotonicity)
  \item $h(\vect{X}) + h(\vect{Y}) -  h(\vect{X}\cap \vect{Y}) - h(\vect{X}\cup \vect{Y})\geq 0$ (submodularity)
        \end{itemize}
      \end{minipage}
    };
  \end{tikzpicture}
  \caption{Polymatroidal axioms. \label{fig:polym}}
\end{figure}

  An \emph{information inequality} is a linear inequality of the form  $c_1h(\vect{X}_1) + \ldots c_n h(\vect{X}_n) \geq 0$, $c_i \in \mathbb R$. The inequality is called a \emph{Shannon} inequality if it follows from the polymatroidal axioms, and \emph{entropic} if it holds true for all entropic vectors. 
  The polymatroidal axioms constitute the basic laws of entropy, but they are incomplete, meaning that there exist entropic information inequalities that do not follow from these axioms.
  The first non-Shannon entropic  inequality,
\begin{equation}\label{eq:zhang}
-I_h(A;B)+2I_h(A;B\mid D)+I_h(A;B\mid C)+I_h(C;D)+I_h(A;D\mid B)+I_h(B;D\mid A)\geq 0,
\end{equation}
was found by Zhang and Yeung \cite{ZhangY98}.
The method of proving that Eq. \eqref{eq:zhang} is entropic was subsequently identified and named as the {Copy Lemma} \cite{dougherty2011nonshannon}, and it has since been employed to uncover numerous other non-Shannon entropic inequalities. 
  

\begin{lemma}[Copy Lemma]
Let $\vect{X},\vect{Y}$ be two sequences of variables that do not share any variables, and let $\vect{Y}'$ be a fresh copy of $\vect{Y}$.
A joint distribution of $\vect{X}\vect{Y}$ can be extended to a joint distribution of $\vect{X}\vect{Y}\vect{Y}'$ such that the marginal distributions on $\vect{X}\vect{Y}$ and $\vect{X}\vect{Y}'$ are identical, and the conditional independence $\pci{\vect{X}}{\vect{Y}}{ \vect{Y}'}$ is satisfied.
%
\end{lemma}
The application of the lemma is simple in principle.
 The joint distribution is extended with new random 
variables that are associated with an identity constraint and a conditional independence constraint.
These constraints are then leveraged 
to derive an information inequality that contains only variables that appear 
in the initial distribution. Since the inequality holds in the extended distribution, it consequently must
also hold in the initial distribution.

The chase, in turn, is an iterative process that, given a set of dependencies $\Sigma \cup\{\tau\}$, starts with a minimal relation $R$ not satisfying $\tau$. Assuming $\tau$ is an EMVD of the form $\vect{X} \twoheadrightarrow \vect{Y}\mid \vect{Z}$, $R$ consists of two tuples $t_0$ and $t_1$ that agree on $\vect{X}$ and disagree on the remaining variables.
At each chase step, one applies some $\sigma \in \Sigma$ to $R$, where $\sigma$ is a dependency that is not satisfied by $R$. For example, if $\sigma$ is an EMVD of the form $\vect{U} \twoheadrightarrow \vect{V}\mid \vect{W}$, one selects two tuples $t,t'\in R$ witnessing that $\sigma$ is not satisfied by $R$. The application of $\sigma$ to $R$ with respect to $t,t'$ then extends $R$ with a new tuple which copies the values of $\vect{UV}$ and $\vect{W}\setminus \vect{UV}$ respectively from $t$ and $t'$, and assigns a fresh value for the remaining variables. If  the chase eventually produces a tuple $t^*$ that agrees on $\vect{XY}$ and $\vect{Z}\setminus \vect{XY}$  respectively with $t_0$ and $t_1$, one may conclude that $\Sigma \models \tau$. 
For more information, the reader may consult standard textbooks in database theory (e.g., \cite{ABLMP21}).

These two techniques, while quite different at the outset, have similarities. The next example demonstrates how the chase can be simulated with the Copy Lemma, using additional inclusion dependencies. 

We now adopt the convention of using boldface letters to denote sets and sequences interchangeably.
If $\vect{X}$ is introduced as a set, we may later on write $\vect{X}$ to denote a sequence listing all elements of $\vect{X}$. Conversely, if $\vect{X}$ is introduced as a sequence, we may subsequently write $\vect{X}$ for the set consisting of all elements listed in $\vect{X}$. Whether $\vect{X}$ refers to a sequence or a set will be always clear from the context. 
 For two  sequences 
 $\vect{X}$ and ${\vect{Y}}$, we write $\vect{X}\vect{\vect{Y}}$ to denote their concatenation. 
 If $f$ is a function whose domain includes every element of a sequence $\vect{X}=(A_1, \dots ,A_n)$, we write $f(\vect{X})$ for the sequence $(f(A_1), \dots ,f(A_n))$.
An \emph{inclusion dependency} (defined for a single relation) is an expression of the form $\vect{X}\subseteq \vect{Y}$, where $\vect{X}$ and $\vect{Y}$ are sequences of distinct variables of the same length. A relation $R$ \emph{satisfies} $\vect{X}\subseteq \vect{Y}$, denoted $R \models \vect{X}\subseteq \vect{Y}$, if for all tuples $t\in R$ there exists a tuple $t'\in R$ such that $t(\vect{X})= t'(\vect{Y})$. 

\begin{example}\label{ex:bchase}
Consider the cycle rule ($n=3$) from Example \ref{ex:cycle}, rewritten in terms of EMVDs as
  \begin{equation*}
    \text{if $B \twoheadrightarrow A \mid D, C \twoheadrightarrow B \mid D, A \twoheadrightarrow C \mid D$, then
    $A \twoheadrightarrow B \mid D$.}
      \end{equation*} 
      The chase with respect to this rule is illustrated in Fig. \ref{fig:chase}. Note that $t_0,t_1$ are the initial tuples; $t_2$ is obtained by applying $A \twoheadrightarrow C \mid D$ to $t_0,t_1$; $t_3$ is obtained by applying $C \twoheadrightarrow B \mid D$ to $t_0,t_2$; and $t_4$ is obtained by applying $B \twoheadrightarrow A \mid D$ to $t_0,t_3$. Since $t_4$ agrees  on $AB$ and $D$ respectively with $t_0$ and $t_1$, we conclude that the rule is sound.
      
  \begin{figure}
  $$
  \begin{array}{ccc}
  \begin{array}{l|llll}
   & {A} & {B}& {C}& {D}\\\cline{2-5}
   t_0& a_0 & b_0 & c_0 & d_0 \\
   t_1 &a_0 & b_1 & c_1 & d_1 \\
    t_2&a_0 & b_2 & c_0 & d_1 \\
   t_3& a_1 & b_0 & c_0 & d_1 \\
    t_4&a_0 & b_0 & c_2 & d_1 \\
  \end{array}
  &
  \quad
  &
  \begin{array}{ll}
  {}\\
  \sigma_0 =ABC D &\hspace{-3mm}\rdelim\}{5}{0mm}[\hspace{1mm}$\subseteq ABCD$] \\
   \sigma_1=AB'C'D' \\
    \sigma_2=AB'' CD' \\
   \sigma_3= A'BCD' \\
    \sigma_4=ABC''D' \\
  \end{array}
    \end{array}
    $$
  \caption{Left: Tuples constructed in the chase example. Right: Inclusion dependencies derived in the Copy Lemma example. \label{fig:chase}}
\end{figure}

  This prodecure can be mirrored using the ``copy lemma''. Consider two inference rules which are clearly sound:
  \begin{enumerate}
  \item If $\vect{XY} \subseteq \vect{UV}, \vect{XZ} \subseteq \vect{UW}, \vect{U}\twoheadrightarrow \vect{V}\mid \vect{W}$, then $\vect{XYZ} \subseteq \vect{UVW}$.
  \item If $\vect{XY} \subseteq \vect{UV}, \vect{XZ} \subseteq \vect{UW}, \vect{U}\twoheadrightarrow \vect{V}\mid \vect{W}, \vect{UVW} \subseteq \vect{XYZ}$, then $\vect{X}\twoheadrightarrow \vect{Y}\mid \vect{Z}$.
  \end{enumerate}

Let $R(ABCD)$ be a relation satisfying the antecedents of the cycle rule above. By the ``copy lemma'', $R$ can be extended to a relation $R'(ABCDB'C'D')$ satisfying the identity constraint $ABCD \subseteq AB'C'D'\land AB'B'D' \subseteq ABCD$ and the independence constraint $A \twoheadrightarrow BCD \mid B'C'D'$. In particular, $R'$ satisfies $\sigma_0= ABCD \subseteq ABCD$ and $\sigma_1 = AB'C'D' \subseteq ABCD$. The first inference rule applied to $A \twoheadrightarrow C \mid D,AC \subseteq AC,AD'\subseteq AD$ produces $\sigma'_2= ACD'\subseteq ACD$. Analogously,  $C \twoheadrightarrow B \mid D,BC \subseteq BC, CD' \subseteq CD$ leads to $\sigma'_3=BCD \subseteq BCD'$, while $B \twoheadrightarrow A \mid D,AB \subseteq AB,BD' \subseteq BD$ yields $\sigma'_4=ABD' \subseteq ABD$. Finally, the second inference rule applied to $AB \subseteq AB,AD' \subseteq AD,A \twoheadrightarrow B\mid D',ABD' \subseteq ABD$ produces the consequent $A \twoheadrightarrow B\mid D$ of the cycle rule. Since $R'$ satisfies $A \twoheadrightarrow B\mid D$, its projection on $ABCD$, which is $R$, must also satisfy the same EMVD.

Fig. \ref{fig:chase} demonstrates the similarities between these applications of the chase and the ``copy lemma''. Note that the relation $R'$ satisfying $\{\sigma_0,\sigma_1,\sigma'_2,\sigma'_3,\sigma'_4\}$ can be further extended to a relation $R''(ABCDA'B'C'D'B''C'')$ satisfying $\{\sigma_0,\sigma_1,\sigma_2,\sigma_3,\sigma_4\}$.
\end{example}

The previous example suggests that the Copy Lemma can also be adapted to contexts other than information theory. Next we will formulate and prove a general version of the Copy Lemma for $K$-relations. 

\subsection{Copy Lemma for $K$-relations}
The Copy Lemma makes use of an identity constraint and a conditional independence constraint. We have already considered conditional independence for $K$-relations. The identity constraint, in turn, can be formulated in terms of marginal identities, introduced in \cite{HKMV18}.
Henceforward, we assume that each variable $A \in \att$ has the whole value domain as its domain, i.e., $\Dom(A)=\val$.
For two sequences of distinct variables $\vect{X}$ and $\vect{Y}$,
a \emph{marginal identity} (MID) is an expression of the form $\vect{X} \approx \vect{Y}$. A $K$-relation $R$ \emph{satisfies} $\vect{X} \approx \vect{Y}$, written $R \models \vect{X} \approx \vect{Y}$, if for all tuples $t\in \tup(\vect{X})$ and $t'\in \tup(\vect{Y})$ such that $t(\vect{X}) = t'(\vect{Y})$, it holds that $R(t)= R(t')$.

We observe that a database relation $R$ satisfies $\vect{X} \approx \vect{Y}$ if and only if it satisfies the inclusion dependencies $\vect{X} \subseteq \vect{Y}$ and $\vect{Y} \subseteq \vect{X}$. A probability distribution $R$ satisfies the marginal identity if and only if its marginal distributions on $\vect{X}$ and $\vect{Y}$ are identical up to pointwise variable renaming. 

The axiomatic properties of marginal identities (Fig. \ref{fig:midaxioms}) are similar to those of inclusion dependencies, except that we need to add one rule for symmetricity.

\begin{figure}[ht]
  \centering
  \begin{tikzpicture}[every node/.style={outer sep=0pt}]
  \def\m{1.4em}
    \node[draw,minimum width=0.7\textwidth,below=0cm of box1.south west,anchor=north west, rounded corners, minimum height=3.8cm] (box3) {
      \hspace{1cm}\begin{minipage}[t][3cm]{.9\textwidth}
        \begin{itemize}
\item[(MID1)] Reflexivity: $A_1,\dots,A_n \approx A_1, \dots ,A_n$.
\item[(MID2)] Symmetry: if $A_1,\dots,A_n \approx B_1, \dots ,B_n$, then $B_1,\dots,B_n \approx A_1, \dots ,A_n$.
\item[(MID3)] Projection and permutation: if $A_1, \dots ,A_n \approx B_1, \dots ,B_n$, then $A_{i_1}, \dots,A_{i_k} \approx B_{i_1}, \dots,B_{i_k}$, where $i_1,\dots ,i_k$ is a sequence of pairwise distinct integers from of $\{1,\dots ,n\}$.
\item[(MID4)] Transitivity: if $A_1, \dots ,A_n \approx B_1, \dots ,B_n$ and $B_1, \dots ,B_n \approx C_1, \dots ,C_n$, then $A_1, \dots ,A_n \approx C_1, \dots ,C_n$. 
        \end{itemize}
      \end{minipage}
    };
  \end{tikzpicture}
  \caption{Axioms of MIDs. \label{fig:midaxioms}}
\end{figure}

\begin{theorem}\label{prop:marg}
Let $K$ be a semiring.
The rules (MID1-MID4) are sound over $K$-relations. They are also complete for MIDs if $K$ is $\kplus$-positive.
\end{theorem}
\begin{proof}
The first statement is obvious. By \cite[Theorem 33]{HannulaV22} the axioms of MIDs are complete over finite probability distributions; the proof can be easily adjusted for $K$-relations if $K$ is $\kplus$-positive.
\end{proof}
We can now generalize the Copy Lemma to $K$-relations.
\begin{lemma}[Copy Lemma for $K$-relations]\label{lem:copy}
Let $\vect{X},\vect{Y}$ be two sequences of variables that do not share any variables, and let $\vect{Y}'$ be a fresh copy of $\vect{Y}$.
Let $K$ be a semiring that does not have divisors of zero. For any $K$-relation $R(\vect{X}\vect{Y})$, there exists a $K$-relation $R'(\vect{X}\vect{Y}\vect{Y}')$, $R'[\vect{X}\vect{Y}]\equiv R$, which satisfies $\pci{\vect{X}}{\vect{Y}}{\vect{Y}'}$ and $\vect{X}\vect{Y} \approx \vect{X}\vect{Y}'$. 
\end{lemma}
\begin{proof}
  Let $\pi$ be the function that maps any $\vect{X}\vect{Y}'$-tuple $t'$ to the $\vect{X}\vect{Y}$-tuple $t$ such that $t'(\vect{X}\vect{Y}')=t(\vect{X}\vect{Y})$.
We define $R'$ as follows: 
\[
R'(t) \coloneqq R(t[\vect{X}\vect{Y}])  R(\pi(t[\vect{X}\vect{Y}']))  c_{R}(t[\vect{X}]).
\]
Recall that $\tup(\vect{Y})=\tup(\vect{Y}')$ since we assumed that the domain of each variable is $\val$.
Thus we observe 
\begin{align*}
R'(t[\vect{X}\vect{Y}]) &=\, \bigkplus_{\substack{t_1\in \tup(\vect{X}\vect{Y}\vect{Y}')\\t_1[\vect{X}\vect{Y}]=t[\vect{X}\vect{Y}]}} R(t_1[\vect{X}\vect{Y}])  R(\pi(t_1[\vect{X}\vect{Y}']))  c_{R}(t_1[\vect{X}])\\
&=\,c_{R}(t[\vect{X}])R(t[\vect{X}\vect{Y}])  \bigkplus_{\substack{t_1\in \tup(\vect{X}\vect{Y}\vect{Y}')\\t_1[\vect{X}\vect{Y}]=t[\vect{X}\vect{Y}]}}   R(\pi(t_1[\vect{X}\vect{Y}'])) \\
&=\,c_{R}(t[\vect{X}])R(t[\vect{X}\vect{Y}])  \bigkplus_{\substack{t_1\in \tup(\vect{X}\vect{Y})\\t_1[\vect{X}]=t[\vect{X}]}}   R(t_1[\vect{X}\vect{Y}]) \\
&=\, c_{R}(t[\vect{X}])R(t[\vect{X}\vect{Y}])  R(t[\vect{X}]) 
=\,  c^*_{R,\vect{X}} R(t[\vect{X}\vect{Y}]),
\end{align*}
by which we obtain $R'[\vect{X}\vect{Y}]\equiv R$, since $c^*_{R,\vect{X}} \neq 0$ by Proposition \ref{prop:notzero}.

Concerning the marginal identity, let $t\in \tup(\vect{X}\vect{Y})$ and $t'\in \tup(\vect{X}\vect{Y}')$ 
be two tuples such that $t=\pi(t')$. 
 Then,
\begin{align*}
 R'(t) &=\,  \bigkplus_{\substack{t_1\in \tup(\vect{X}\vect{Y}\vect{Y}')\\t_1[\vect{X}\vect{Y}]=t}}
R(t_1[\vect{X}\vect{Y}])  R(\pi( t_1[\vect{X}\vect{Y}']))  c_{R}(t_1[\vect{X}])\\
&=\,
 c_{R}(t[\vect{X}])  R(t)
  \bigkplus_{\substack{t_1\in \tup(\vect{X}\vect{Y}\vect{Y}')\\t_1[\vect{X}\vect{Y}]=t}}
R(\pi( t_1[\vect{X}\vect{Y}'])) 
=\, c_{R}(t[\vect{X}])  R(t)   R(t[\vect{X}]) \\
&=\, R(t)  c_{R}(t[\vect{X}])\bigkplus_{\substack{t_1\in \tup(\vect{X}\vect{Y})\\t_1[\vect{X}]=t[\vect{X}]}}
R(t_1[\vect{X}\vect{Y}]) =\, R(\pi( t'))  c_{R}(t[\vect{X}])\bigkplus_{\substack{t_1\in \tup(\vect{X}\vect{Y}\vect{Y}')\\t_1[\vect{X}\vect{Y}']=t'}}
R(t_1[\vect{X}\vect{Y}]) \\
&=\,\bigkplus_{\substack{t_1\in \tup(\vect{X}\vect{Y}\vect{Y}')\\t_1[\vect{X}\vect{Y}']=t'}}
R(t_1[\vect{X}\vect{Y}])  R(\pi( t_1[\vect{X}\vect{Y}'])) c_{R}(t_1[\vect{X}])\\
&=\, R'(t').
\end{align*}
Considering the CI  $\pci{\vect{X}}{\vect{Y}}{\vect{Y}'}$, letting $t\in \tup(\vect{X}\vect{Y}\vect{Y}')$, we obtain
\begin{align*}
&\, R'(t[\vect{X}\vect{Y}])  R'(t[\vect{X}\vect{Y}']) \\
=&\,  
\bigkplus_{\substack{t_1 \in \tup(\vect{X}\vect{Y}\vect{Y}')\\t_1[\vect{X}\vect{Y}]=t[\vect{X}\vect{Y}]}}
R(t_1[\vect{X}\vect{Y}])  R(\pi(t_1[\vect{X}\vect{Y}']))  c_{R}(t_1[\vect{X}])\\
&\,\bigkplus_{\substack{t_1 \in \tup(\vect{X}\vect{Y}\vect{Y}')\\t_1[\vect{X}\vect{Y}']=t[\vect{X}\vect{Y}']}}
R(t_1[\vect{X}\vect{Y}])  R(\pi(t_1[\vect{X}\vect{Y}']))  c_{R}(t_1[\vect{X}])
\\
=&\,  
R(t[\vect{X}\vect{Y}])R(t[\vect{X}])c_{R}(t[\vect{X}])
R(t[\vect{X}])  
R(\pi(t[\vect{X}\vect{Y}'])) c_{R}(t[\vect{X}])
\\
=&\,  
R(t[\vect{X}\vect{Y}])  R(\pi(t[\vect{X}\vect{Y}']))  c_{R}(t[\vect{X}])
\bigkplus_{\substack{t_1 \in \tup(\vect{X}\vect{Y}\vect{Y}')\\t_1[\vect{X}]=t[\vect{X}]}}
R(t_1[\vect{X}\vect{Y}])  R(\pi(t_1[\vect{X}\vect{Y}']))  c_{R}(t_1[\vect{X}])
\\
=&\, R'(t)  R'(t[\vect{X}]).
\end{align*}
This concludes the proof of the lemma.
\end{proof}

\subsection{Axiomatic Application of Copy Lemma}
In this section we consider how the Copy Lemma can be systematically
exploited in conjunction with an axiom system.
Fig. \ref{fig:copy} presents the Copy Lemma as an axiom. 
The variables that appear in the fresh copy $\vect{Y}'$ are called \emph{existential}, since these variables
can be thought of as existentially quantified.
 A proof that uses the Copy Lemma should not include existential variables in its conclusion. 
 A formal definition of a valid proof involving the Copy Lemma is given below.
 
\begin{figure}[ht]
  \centering
  \begin{tikzpicture}[every node/.style={outer sep=0pt}]
  \def\m{1.4em}
    \node[draw,rounded corners, minimum width=0.5\textwidth, minimum height=1cm] (box1) {
            \hspace{2mm}
      \begin{minipage}[t][.4cm]{.9\textwidth}
   \begin{itemize}
\item 
$\pci{\vect{X}}{\vect{Y}}{\vect{Y}'}\text{ and }\vect{X}\vect{Y} \approx \vect{X}\vect{Y}'$, where $\vect{Y}'$ is a fresh copy of $\vect{Y}$. 
   \end{itemize}
  \end{minipage}
    };
  \end{tikzpicture}
  \caption{Copy Lemma as an axiom. \label{fig:copy}}
\end{figure}

A set $\mathfrak{A}$ of inference rules of the form \eqref{eq:rule}
 is called an \emph{axiomatization}. The axiomatization $\mathfrak{A}$ is \emph{sound} for  $K$-relations if all of its rules are sound for  $K$-relations.
Let $\Sigma \cup \{\tau\}$ be a set of formulae, and let 
$\frakA$ be an axiomatization. We say
that $\tau$ is \emph{provable from $\Sigma$ by $\frakA$ and Copy Lemma}, 
 written $\Sigma \vdash_{\frakA,CL} \tau$, if there exists a sequence 
of formulae $(\tau_1, \dots ,\tau_k)$ (called a \emph{proof} of $\tau$ from $\Sigma$)  such that 
\begin{enumerate}
  \item $\tau=\tau_k$;
  \item $\tau$ contains no existential variables; and
  \item for $i\in [k]$, 
\begin{itemize}
   \item $\tau_{i}$ is from $\Sigma$ or obtained from $\Sigma \cup\{\tau_1, \dots ,\tau_{i-1}\}$ using
    one of the rules of $\frakA$, or
  \item  $\tau_i$ and $\tau_{i+1}$ (or, resp. $\tau_{i-1}$ and $\tau_{i}$) are obtained by Copy Lemma, and the fresh variables introduced do not appear 
   in $\Sigma \cup \{\tau_{1}, \dots ,\tau_{i-1}\}$ (resp. $\Sigma \cup \{\tau_{1}, \dots ,\tau_{i-2}\}$).
  \end{itemize}
  \end{enumerate}

 The Copy Lemma can be safely attached to any sound axiom system provided that certain invariance principles hold. 
 We say that a
formula 
  $\sigma$ over a variable set $\vect{X}$ is \emph{invariant under equivalence} w.r.t. $K$ if for all $K$-relations $R(\vect{V})$ and $R'(\vect{V})$ such that $\vect{X}\subseteq\vect{V}$ and $R \equiv R'$, $R \models \sigma$ if and only if $R'\models \sigma$. The formula is \emph{invariant under projection} if for all $K$-relations $R(\vect{W})$ and sets $\vect{V}$ where $\vect{X}\subseteq\vect{V} \subseteq \vect{W}$, $R \models \sigma$ if and only if $R[\vect{V}]\models \sigma$. 
We say that an axiomatization $\mathfrak A$ is \emph{invariant under equivalence} (resp. \emph{projection}) w.r.t. $K$ if it contains only formulae that are invariant under equivalence (resp. projection) w.r.t. $K$.

\begin{theorem}\label{thm:sound}
Let $K$ be a semiring that does not have divisors of zero. 
Let $\frakA$ be an axiomatization that is sound for $K$-relations and invariant under equivalence and projection w.r.t. $K$.
For a set of formulae $\Sigma \cup \{\tau\}$,
 $\Sigma \vdash_{\frakA,CL} \tau$ implies $\Sigma \models_K \tau$.
\end{theorem}
\begin{proof}
Let $\tau_1, \dots ,\tau_k$ be a proof of $\tau$ from $\Sigma$.
Let $\vect{V}_0$ be the set of non-existential variables appearing in $\Sigma \cup\{\tau_1, \dots ,\tau_k\}$.
For $i \geq 1$, let $\vect{V}_i$, be the extension of $\vect{V}_0$ with all existential variables that appear in 
$\{\tau_1, \dots ,\tau_i\}$. 
Let $R_0(\vect{V}_0)$ be a $K$-relation satisfying $\Sigma$. 
We need to show that $R_0$ satisfies $\tau$. For this,
by invariance under equivalence and projection, it suffices to prove inductively that
 for each $i \in [k]$ there is a 
$K$-relation $R_i(\vect{V}_i)$, $R_i[\vect{V}_0]\equiv R_0$,
satisfying $\Sigma \cup \{\tau_1, \dots ,\tau_i\}$.

The base step is clear. For the inductive step, suppose $R_{i}(\vect{V}_{i})$, $R_i[\vect{V}_0]\equiv R_0$, is a $K$-relation satisfying $\Sigma \cup \{\tau_1, \dots ,\tau_{i}\}$.
Assume first that $\tau_{i+1}$ is from $\Sigma$ or obtained from $\Sigma\cup\{\tau_1, \dots ,\tau_{i}\}$ using a rule from $\frakA$. Then $\vect{V}_{i}=\vect{V}_{i+1}$, and by selecting $R_{i+1}(\vect{V}_{i+1})$ 
as $R_{i}(\vect{V}_{i})$ the induction claim follows readily by induction hypothesis and soundness of $\frakA$. 

Suppose then $\tau_{i+1}$ is obtained
by Copy Lemma. Without loss of generality we may assume that 
$\tau_{i+1}$ and $\tau_{i+2}$ are the results of applying the lemma.
Suppose these formulae are of the form $\pci{\vect{X}}{\vect{Y}}{\vect{Y}'}$ and $\vect{X}\vect{Y} \approx \vect{X}\vect{Y}'$.
In particular, we have $\vect{X}\vect{Y} \subseteq \vect{V}_{i}$ and $\vect{Y}' \cap \vect{V}_{i}=\emptyset$, where $\vect{Y}'$ is the fresh copy of $\vect{Y}$. Let $\vect{W} = \vect{V}_{i} \setminus \vect{X}\vect{Y}$.
By Lemma \ref{lem:copy},
there exists a relation $R_{i+1}(\vect{V}_{i+1})$, $R_{i+1}[\vect{V}_{i}]\equiv R_{i}$, which 
satisfies  $\pci{\vect{X}}{\vect{Y}\vect{W}}{\vect{Y}'\vect{W}'}$ and $\vect{X}\vect{Y}\vect{W} \approx \vect{X}\vect{Y}'\vect{W}'$. Consequently, by Theorems \ref{thm:sound:pci} and \ref{prop:marg}, $R_{i+1}$ 
must also satisfy 
 $\tau_{i+1}$ and 
$\tau_{i+2}$. Thus, by invariance under equivalence and projection, $R_{i+1}$ satisfies $\Sigma \cup \{\tau_1, \dots ,\tau_{i+1}\}$. Moreover, applying induction hypothesis together with Lemmata \ref{lem:phokion} and \ref{lem:equiv} we  obtain that 
$R_{i+1}[\vect{V}_0] = R_{i+1}[\vect{V}_{i}][\vect{V}_0] \equiv R_{i}[\vect{V}_0] \equiv R_0$, and thus $R_{i+1}[\vect{V}_0] \equiv R_0$  because by Lemma \ref{lem:equiv} the equivalence relation is transitive for any $K$ that lacks zero divisors.
This concludes the proof.
%
\end{proof}
 
Using the previous theorem one can expand Example \ref{ex:bchase} to a proof system that includes the Copy Lemma,
and is sound and complete for EMVDs over (possibly infinite) database relations (cf. \cite{HannulaK16} where a principle similar to the Copy Lemma is applied).
In the past, an analogous method involving fresh variables has also been employed to axiomatize functional and inclusion dependencies
 \cite{mitchell83}.

Theorem \ref{thm:sound} can also be applied to find entropic information inequalities.
Consider the Copy Lemma (Fig. \ref{fig:copy}), the polymatroidal axioms (Fig. \ref{fig:polym}), and the axioms of MIDs (Fig. \ref{fig:midaxioms}). Consider also 
non-negative combinations of information inequalities together with a rule  concerning interaction with MIDs (Fig. \ref{fig:infaxioms}). 
Clearly, these rules form an axiomatization that is sound and invariant under equivalence and projection for $\mathbb R_{\geq 0}$-relations, assuming $h$ is interpreted as the Shannon entropy associated with the probability distribution $\frac{1}{|R[\emptyset]|} R$.
In particular, in this context the CI  $\pci{\vect{X}}{\vect{Y}}{\vect{Z}}$ is equivalent to the information inequality
$-I_h(\vect{Y};\vect{Z}\mid \vect{X}) \geq 0$, where $I_h(\vect{Y};\vect{Z}\mid \vect{X})$ is the conditional mutual information defined in \eqref{eq:condmutual}.
The following example presents a purely syntactic derivation of the Zhang-Yeung inequality \eqref{eq:zhang}
within this proof system. We now drop  the subscript $h$ from conditional mutual information to avoid unnecessary repetition.

\begin{figure}[ht]
  \centering
  \begin{tikzpicture}[every node/.style={outer sep=0pt}]
  \def\m{1.4em}
          \node[draw,rounded corners, minimum height=3cm, minimum width=0.6\textwidth,below=0cm of box1.south west,anchor=north west] (box4) {
      \begin{minipage}[t][2.3cm]{1\textwidth}
        \begin{itemize}
        \item If $A_1,\dots,A_n \approx B_1, \dots ,B_n$, then $h(A_1\dots A_n) - h(B_1\dots B_n) \geq 0$ (interaction).
\item If $c_1 h(\vect{X}_1) + \dots +c_nh(\vect{X}_n) \geq 0$ and $d > 0$, then $d(c_1 h(\vect{X}_1) + \dots + c_nh(\vect{X}_n)) \geq 0$ (positive scalar multiplication).
\item If $c_1 h(\vect{X}_1) + \dots +c_nh(\vect{X}_n) \geq 0$ and $d_1 h(\vect{X}_1) + \dots +d_nh(\vect{X}_n) \geq 0$, \\
then $(c_1+d_1) h(\vect{X}_1) + \dots +(c_n+d_n)h(\vect{X}_n) \geq 0$ (addition).
        \end{itemize}
      \end{minipage}
    };
  \end{tikzpicture}
  \caption{Interaction and positive combinations. \label{fig:infaxioms}}
\end{figure}

\begin{example}
  The following formal proof of the Zhang-Yeung inequality is based on the presentations in \cite{laszlo22,Suciu23}.
    \begin{itemize}
        \item[(1)] $-I(CD; D' \mid AB) \geq 0$ 
    \item[(2)] $ABD \approx ABD'$ 
    \makebox(0,0){\put(31,2.2\normalbaselineskip){%
               $\left.\rule{0pt}{1.1\normalbaselineskip}\right\}$ (Copy Lemma, $D'$ existential)}}
    \item[(3)]  $I(C; D\mid D') + I(C; D'\mid B ) + I(D; D'\mid A)
    + I(D; D'\mid BC ) + I(C; D'\mid B ) + I(C; D'\mid A)
    + I(C; D'\mid ABD ) + I(A; B \mid CD') + I(A; B \mid DD') + I(A; D'\mid BCD ) + I(B ; D'\mid CD) \geq 0$ (addition and polymatroidal axioms)
    \item[(4)]  $-I(A;B)+I(A;B\mid C)+I(A;B\mid D)+I(C;D)+ I(A;B\mid D')+I(D';A\mid B)+I(D';B\mid A)+3I(CD;D'\mid AB) \geq 0$ (step 3 rewritten)
    \item[(5)] $-3I(CD; D' \mid AB) \geq 0$ (positive scalar multiplication and step 1)
    \item[(6)] $-I(A;B)+I(A;B\mid C)+I(A;B\mid D)+I(C;D)+ I(A;B\mid D')+I(D';A\mid B)+I(D';B\mid A) \geq 0$ (addition, and steps 3 and 5)
    \item[(7)] $ABD' \approx ABD$ (symmetry and step 2)
    \item[(8)] $AD \approx AD'$ (projection and permutation, and step 2)
    \item[(9)] $BD \approx BD'$ (projection and permutation, and step 2)  
    \item[(10)] $h(ABD') - h(ABD) \geq 0$ (interaction and step 7)
    \item[(11)] $h(AD) - h(AD') \geq 0$ (interaction and step 8)
    \item[(12)] $h(AD) - h(AD') \geq 0$ (interaction and step 9)
    \item[(13)] $3h(ABD') - 3h(ABD) + 2h(AD) - 2h(AD') + 2h(AD) - 2h(AD')\geq 0$ (addition, positive scalar multiplication, and steps 10-12)
    \item[(14)] $-I(A;B)+I(A;B\mid C)+I(A;B\mid D)+I(C;D)+ I(A;B\mid D)+I(D;A\mid B)+I(D;B\mid A) \geq 0$ (addition, and steps 6 and 13)  
  \end{itemize} 
  Since the existential variable $D'$ is not present at the last step, the sequence forms a proof of the Zhang-Yeung inequality. 
  This concludes the example.
\end{example}

Matu\^{s} proved that the almost entropic region $\aent{n}$ is not 
polyhedral by showing that it is bordered by 
a sequence of inequalities ($k\geq 1$)
\begin{equation}\label{eq:matus}
  k \big(-I(A;B) +I(C;D) + I(A;B\mid C)\big) + I(A;D\mid B) + \frac{k(k+1)}{2}I(B;D\mid A) + \frac{k(k+3)}{2}I(A;B\mid D) \geq 0
\end{equation}
 displaying quadratic behaviour \cite{Matus07}. 
The 
Zhang-Yeung inequality \eqref{eq:zhang} in particular is obtained from Eq. \eqref{eq:matus} by plugging in $k=1$. Analogously to the previous example, one can show that Matu\^{s}' inductive proof of the inequalities \eqref{eq:matus} corresponds to a formal proof that combines nested application of Copy Lemma (Fig. \ref{fig:copy}) together with the rules of Figs. \ref{fig:midaxioms} and \ref{fig:infaxioms}.\footnote{
  \cite[Theorem 2]{Matus07}, which states a sequence of inequalities over five variables,
  can  be easily replicated within this proof system. From this,
  \eqref{eq:matus} (i.e., \cite[Corollary 2]{Matus07}) follows by a variable substitution; this step can simulated by taking one extra copy in the beginning of the proof.
}

There have also been proposals to extend and modify the Copy Lemma in order to find 
more non-Shannon inequalities. Nonetheless, the Copy Lemma in its classical form still appears strong enough to be able to prove all known non-Shannon entropic inequalities \cite{GurpinarR19}. Whether or not the Copy Lemma can really prove all entropic inequalities is also related to the so-called \emph{information inequality problem}, which  is to determine whether a given information inequality is entropic. This problem is known to be co-r.e. \cite{KhamisK0S20}, but its decidability is an open question.  Consequently, if the Copy Lemma can be extended to an axiom system able to prove all entropic inequalities, this would entail that
 the information inequality problem is decidable.

\section{Conclusion}
We have studied axiomatic and decomposition properties of conditional independence over $K$-relations. For positive and multiplicatively cancellative $K$, 
we showed that (i) conditional independence corresponds to lossless-join decompositions, (ii) the semigraphoid axioms of conditional independence are sound, and (iii) saturated conditional independence and functional dependence have a sound and complete axiom system, mirroring the sound and complete axiom system of MVDs and FDs.
To demonstrate possible applications, we provided an example data table that admits a lossless-join decomposition only when one of its variables is reinterpreted as a semiring annotation. Finally, we  considered a model-theoretic approach to study the relationships between different CI semantics, and showed that the Copy Lemma, used in information theory to derive non-Shannon inequalities, can be extended to $K$-relations whenever $K$ is positive.

The questions of the axiomatic characterization \cite{parker:1980,studeny:1993,Studeny95} and decidability \cite{herrmann95, Li23} of the CI implication problem have been answered in the negative in different frameworks. 
Having identified positivity and multiplicative cancellativity as the fundamental semiring properties for the notion of conditional independence, we may now ask whether these negative results extend to any $K$ with these characteristics. 
Furthermore, after demonstrating that the Copy Lemma can emulate the chase algorithm in database theory, it would also be worthwhile to investigate potential applications of this lemma in other semiring contexts.
\bibliography{biblio}

\appendix

\section{Embeddings}\label{sect:embeddings}

We commence this section with two simple propositions.
 \begin{proposition}\label{prop:natordplus}
Any naturally ordered semiring is $\kplus$-positive.
\end{proposition}
\begin{proof}
Suppose $\leq$ is the natural order of $(K,\kplus,\ktimes,0,1)$. Let $a,b\in K$ be such that $a+b =0$. Then by definition of natural order, $a\leq 0$ and $b\leq 0$.
Clearly $0\leq c$ for all $c \in K$. Hence we conclude by antisymmetry that $a=0$ and $b=0$.  
\end{proof}

\begin{proposition}\label{prop:multdiv}
Any multiplicatively cancellative semiring has no divisors of zero.
\end{proposition}
\begin{proof}
If $ab=0$ and $b\neq 0$, then multiplicative cancellativity and $ab=0b$ implies $a=0$.
\end{proof}

Let $\tau$ be a first-order vocabulary consisting of function and relation symbols (constant symbols can be viewed as $0$-ary function symbols). We write 
write $\ar(\alpha)$ for the \emph{arity} of a symbol $\alpha \in \tau$. 
Given a $\tau$-structure $\calM$ and an element $\alpha$ from $\tau$, we write $\alpha^{\calM}$ for the interpretation of $\alpha$ in $\calM$.
Let $\calA$ and $\calB$ be two $\tau$-structures with domains $A$ and $B$.
We call $\calA$ a \emph{submodel} of $\cal{B}$, written $\calA \subseteq \calB$, if 
$A \subseteq B$, and the interpretation of every function symbol and relation symbol in $\tau$ is inherited from $\calB$; i.e., for each $\alpha \in \tau$, $\alpha^{\calA}$ is the restriction $\alpha^{\calB}\upharpoonright A^k$ of $\alpha^{\calB}$ to $A^k$.  We say that $\calA$ and $\calB$ are \emph{isomorphic}, written $\calA \cong \calB$, if there exists a
bijection (called an \emph{isomorphism} between $\calA$ and $\calB$) $\pi:A \to B$ such that 
\begin{itemize}
  \item $\pi(f^{\calA}(a_1, \dots ,a_{\ar(f)}))=f^{\calB}(\pi(a_1), \dots ,\pi(a_{\ar(f)}))$, for all  function symbols $f\in \tau$ and elements $a_1, \dots ,a_k\in A$, and
  \item $(a_1, \dots , a_{\ar(R)})\in R^{\calA} \iff (\pi(a_1), \dots , \pi(a_{\ar(R)}))\in R^{\calB}$, for all  relation symbols $R\in \tau$ and elements $a_1, \dots ,a_k\in A$.
\end{itemize}
We say that $\calA$ \emph{embeds in} $\calB$, written $\calA \preccurlyeq \calB$, if $\calA$ and some submodel of $\calB$ are isomorphic.

The embedding results that follow are proven using the concept of a congruence. We write $[n]$ for the set of integers $\{1, \dots, n\}$.
A \emph{congruence} on a structure $S=(A,(R_i)_{i\in [n]}, (f_j)_{j\in [m]})$ is an equivalence relation $\sim$ that preserves all relations $R_i$ and functions $f_j$, i.e.,  
  \begin{enumerate}
  \item for $i \in [n]$ and $a_1 \sim b_1, \ldots ,a_{\ar(R_i)}\sim b_{\ar(R_i)}$, $(a_1, \ldots, ,a_{\ar(R_i)}) \in R_i$ implies\\ $(b_1, \ldots, ,b_{\ar(R_i)}) \in R_i$; and
  \item for $j \in [m]$ and $a_1 \sim b_1, \ldots, a_{\ar(f_i)}\sim b_{\ar(f_i)}$, $f_j(a_1, \ldots, ,a_{\ar(f_i)}) \sim f_j(b_1, \ldots, ,b_{\ar(f_i)})$.
  \end{enumerate}
  The structure $S /\sim \,\coloneqq (A/\sim, (R_i/\sim)_{i\in [n]}, (f_j/\sim)_{j\in [m]})$ is called the \emph{quotient} of $S$ w.r.t. $\sim$.

Lemma \ref{prop:embedding} now expands a semiring with multiplicative inverses analogously to the construction of positive rationals as fractions of natural numbers, and Lemma \ref{prop:embedding2} expands a semifield  with additive inverses analogously to the construction of rationals as pairs of positive rationals.
 Lemma \ref{lem:semi-cons}  is then obtained as a consequence of Propositions \ref{prop:natordplus} and \ref{prop:multdiv}, and Lemmata \ref{prop:embedding} and \ref{prop:embedding2}.

  \possemi*
\begin{proof}
  Let $(K, \kplus, \ktimes, 0, 1)$ be a positive multiplicatively cancellative semiring. Recall that $K$ has a preorder defined as $a \leq b \colon \iff \exists c(a\kplus c=b)$. 
Define a structure 
  $S = (K \times K\setminus \{0\}, \ktimes_S, 0_S, 1_S, \leq_S)$ by 
  \begin{itemize}
    \item $(a,b) \kplus_S (c,d) \coloneqq (ad \kplus bc, bd)$,
    \item $(a,b) \ktimes_S (c,d) \coloneqq (ac, bd)$,
    \item $0_S \coloneqq (0,1)$,
    \item $1_S \coloneqq (1,1)$,
  \end{itemize}
  Note that that $S$ is not a semiring, since multiplication does not distribute over addition. 
  Define an equivalence relation $\sim$ on $K \times K\setminus\{0\}$ such that $(a,b)\sim (c,d)$ if $ad = bc$. It is straightforward to show, without using multiplicative cancellativity, that
  $\sim$ is a congruence on $S$.   
  We claim that the quotient of $S$ w.r.t. $\sim$ is a positive semifield. 
  
  It is straightforward to verify that  $S/\sim$ is a semifield. To show that $S/\sim$ is positive, we consider first $\kplus$-positivity. Suppose $(a,b)\kplus_S (c,d) \sim 0_S$, in which case we have $ad \kplus bc=0$. Since $K$ is positive and $b\neq 0 \neq d$, it must be the case that $a = 0 = c$, i.e., $(a,b)\sim 0_S \sim (c,d)$. Let us then consider divisibility of zero. Suppose $(a,b)\ktimes_S (c,d) \sim 0_S$, in which case we have $ad=0$. Since $K$ does not have divisors of zero, one of $a$ or $c$ must be zero, i.e., $(a,b)\sim 0_S$ or  $ (c,d)\sim 0_S$.  
  
  We note that $a \mapsto [(a,1)]$, where $[(a,1)]$ is the equivalence class of $(a,1)$,  is clearly an injection from $K$ to $S/\sim$. Hence $K$ embeds in $S/\sim$. This concludes the proof of the first statement of the lemma.
   
 Next we show that $S/\sim$ is additively cancellative whenever $K$ is also additively cancellative. For this, we need to show that $(a,b) \sim (a',b')$ follows from $(a,b) \kplus_S (c,d) \sim (a',b') \kplus_S (c,d)$. Indeed, by definition the latter can be rewritten as $ab'dd\kplus bb'cd = a'bdd \kplus bb'cd$, which in turn yields $(a,b) \sim (a',b')$ by additive and multiplicative cancellativity. This concludes the proof of the second statement.
  
  Finally, we need to show that $S/\sim$ is naturally totally ordered whenever $K$ is naturally totally ordered. To do so, we first define an ordering relation
   \[(a,b) \leq_S (c,d) \colon\iff a  d \leq b  d,\]
   where $\leq$ is the natural order of $K$. Observe that $\sim$ is a congruence also on the extension of $S$ with $\leq_S$.
   For this, note that $(a,b) \leq_S (c,d)$ and $(a,b)\sim (a',b')$ implies $ad \leq bc$ and $ab'=a'b$. Since vacuously $b'\geq 0$, we obtain $adb' \leq bcb'$ and hence $a'db \leq bcb'$. By multiplicative cancellativity we conclude that $a'd \leq b'c$, i.e., $(a',b') \leq_S (c,d)$. A symmetric argument  shows that $(a,b) \leq_S (c,d)$ and $(c,d)\sim (c',d')$ implies $(a,b) \leq (c',d')$.

 We now claim that $\leq_S /\sim $ is defined by Eq. \eqref{eq:nat}.
  Recall that $K$ is multiplicatively cancellative, and $\ktimes$ is monotone under $\leq$.
 Hence, if $a,b\in K$ and $d \in K\setminus\{0\}$, then  $(a,d) \leq_S (b,d)$ if and only if $a \leq b$. Since $\leq$ is the natural order of $K$, this is holds exactly when there exists $c\in K$ such that $a+c =b$. Again, by positivity and multiplicative cancellativity of $K$ this is equivalent to $(a+c)d^2=bd^2$, which is $(a,d)\kplus_S (c,d) = (b,d)$ rewritten.  From this, the claim follows.

Note that $\leq_S /\sim $ is antisymmetric. Suppose $(a,b) \leq_S (c,d) \leq_S (a,b)$. Then $ad\leq bc \leq ad$, meaning that $ad=bc$ by antisymmetry of $\leq$. This entails that $(a,b) \sim (c,d)$.

We conclude that $S /\sim $ is naturally ordered by $\leq_S /\sim $.
It is clear that $\leq_S/\sim$ is a total order, because $\leq$ is total on $K$. Therefore, $S/\sim$ is naturally totally ordered whenever $K$ is naturally totally ordered. This concludes the proof of the third statement.
\end{proof}

  \begin{lemma}\label{prop:embedding2}
  Any naturally totally ordered additively cancellative semifield embeds in a totally ordered field.
\end{lemma}

\begin{proof}
  Let $S=(S, \kplus, \ktimes, 0, 1,\leq)$ naturally totally ordered additively cancellative semifield. 
  Define a structure 
  $F = (S\times S, \kplus_F, \ktimes_F, 0_F, 1_F, \leq_F)$ by 
  \begin{itemize}
    \item $(a,b) \kplus_F (c,d) \coloneqq (a\kplus c, b \kplus d)$,
    \item $(a,b) \ktimes_F (c,d) \coloneqq (ac \kplus bd, ad \kplus bc)$,
    \item $0_F \coloneqq (0,0)$,
    \item $1_F \coloneqq (1,0)$,
    \item $(a,b) \leq_F (c,d) \colon\iff a \kplus d \leq b \kplus c$.
  \end{itemize}
  Define an equivalence relation $\sim$ on $S\times S$ such that $(a,b) \sim (c,d)$ if and only if
$a \kplus d = b \kplus c$. Then, $\sim$ is a congruence on $F$. We prove the cases
for multiplication  and inequality, leaving addition for the reader.

For multiplication, it suffices to prove that $(a,b) \kplus_F (c,d) \sim (a',b') \kplus_F (c,d)$ assuming $(a,b)\sim (a',b')$.
Indeed, we have
\begin{align*}
  (a,b) \kplus_F (c,d) \sim (a',b') \kplus_F (c,d)
  &\iff 
  (ac \kplus bd , ad \kplus bc) \sim (a'c \kplus b'd , a'd \kplus b'c)\\
  &\iff (a\kplus b')c \kplus (a'\kplus b)d = (a' \kplus b)c \kplus (a\kplus b')d,
\end{align*}  
which holds by the assumption.

For inequality, it suffices to show that $(a,b) \leq_F (c,d)$ if and only if $(a',b')\leq_F (c,d)$, assuming $(a,b)\sim (a',b')$.
Using assumption,  additive cancellativity, and additive monotony, we obtain
\begin{align*}
  (a,b) \leq_F (c,d) &\iff a \kplus d \leq b \kplus c \iff a \kplus d \kplus b' \leq b \kplus c \kplus b' \\
  &\iff a' \kplus d \kplus b \leq b \kplus c \kplus b'
  \iff a' \kplus d \leq c \kplus b'\iff (a',b') \leq_F (c,d).
\end{align*}

We conclude that $\sim$ is a congruence on $F$. It is now easy to verify that the quotient of $F$ w.r.t. $\sim$ is a totally ordered field.
We consider only multiplicative inverses and monotony of non-negative multiplication, leaving the 
remaining properties for the reader.

For multiplicative inverses, using the natural total order, for any $(a,b)$ from  $S\times S$
one finds a pair of $S\times S$ of the form $(c,0)$ (or $(0,c)$, resp.) such that $(a,b) \sim (c,0)$ ($(a,b) \sim (0,c)$, resp.).
Then, taking the multiplicative inverse $c^{-1}$ of $c$, we observe $(c,0) \ktimes_F (c^{-1},0) = (0,c) \ktimes_F (0,c^{-1}) = (1,0)$.

For monotony of non-negative multiplication, 
 assume $(a,b) \geq_F (0,0)$ and $(c,d)\geq_F (0,0)$, that is,
$a \geq b$ and $c \geq d$. Since $\leq$ is the natural order of $S$, we find
$e,f\in S$ such that $a = b \kplus e$ and $c = d\kplus f$. We have
\begin{align*}
  (a,b) \ktimes_F (c,d) \geq_F (0,0) &\iff (ac \kplus bd, ad \kplus bc) \geq_F (0,0)\\
  &\iff (b\kplus e)(d\kplus f) \kplus bd \geq (b \kplus e)d \kplus b(d\kplus f)\\
  &\iff ef \geq 0,
\end{align*}
which holds true because $\geq$ is the natural order of $S$. 

Clearly, $a \mapsto [(a,0)]$ is an injection from $S$ to $F/\sim$. This concludes the proof.
\end{proof}

\section{Basic properties of $K$-relations}\label{sect:keq}

\phokion*
\begin{proof}
For sake of completeness, we repeat from \cite{Atserias20} the proof of the first statement. The inclusion $R[\vect{Y}]' \subseteq R'[\vect{Y}]$ holds clearly for all semirings. For the converse direction, if $t$ belongs to $R'[\vect{Y}]$, then we find $t'\in R'$ such that $t'[\vect{Y}]=t$, in which case by $\kplus$-positivity, $t$ must belong to $R[\vect{Y}]'$.

The second statement follows by 
\[
R[\vect{Y}][\vect{Z}](u) =\bigkplus_{\substack{v \in \tup(\vect{Y})\\ v[\vect{Z}]=u}} R[\vect{Y}](v) 
= \bigkplus_{\substack{v \in \tup(\vect{Y})\\ v[\vect{Z}]=u}} \bigkplus_{\substack{w \in \tup(\vect{X})\\ w[\vect{Y}]=v}} R(w)
= \bigkplus_{\substack{w \in \tup(\vect{X})\\ w[\vect{Z}]=u}} R(w) =  R[\vect{Z}](u),
\]
obtained by applying Eq. \eqref{eq:marg} repeatedly.
\end{proof}
In particular, by using Eq. \eqref{eq:marg} as the definition of the marginal, it becomes evident that the second statement holds without the assumption of $K$ being positive.

\simple*
\begin{proof}
  \begin{enumerate}
    \item  Let $a,b\in K\setminus\{0\}$ be such that $aR = bR'$. 
    Then, for any $\vect{W}$-tuple $t$ it holds that
    \[
      aR(t) = \bigkplus_{\substack{t'\in \tup(\vect{V})\\t'[\vect{W}]=t}}aR(t') = \bigkplus_{\substack{t'\in \tup(\vect{V})\\t'[\vect{W}]=t}}bR'(t') = b R'(t). 
    \]
    \item Let $a,b,c,d\in K\setminus\{0\}$ be such that $aR = bR'$ and $cR' = dR''$. Then it is easy to see that
    $acR = cdR''$. Since $K$ has no divisors of zero, $ac, cd \in K\setminus \{0\}$. Thus $R\equiv R''$.
  \end{enumerate}
\end{proof}

\section{Graphoid axioms}\label{sect:graphoid}

We will use the following helping lemma in the proof of Theorem \ref{thm:sound:pci}.
\begin{restatable}{lemma}{poslemma}\label{lem:poslemma}
Let $R(\vect{X})$ be a $K$-relation, where $K$ is $\kplus$-positive. Let $t$ be a tuple of $R$, and let $\vect{Y},\vect{Z}$ be variable sets such that $\vect{Z}\subseteq \vect{Y}\subseteq \vect{X}$.
Then,  $t[\vect{Y}] \in R[\vect{Y}]'$ implies $t[\vect{Z}] \in R[\vect{Z}]'$. 
\end{restatable}
\begin{proof}
By Lemma \ref{lem:phokion}, $R[\vect{Y}][\vect{Z}]=R[\vect{Z}]$.
Using Eq. \eqref{eq:marg} we have
\[
R(t[\vect{Z}]) = \bigkplus_{\substack{u\in \tup(\vect{Y})\\u[\vect{Z}]=t[\vect{Z}]}} R(u) = R(t[\vect{Y}] ) \kplus \bigkplus_{\substack{u\in \tup(\vect{Y})\\u[\vect{Z}]=t[\vect{Z}]\\u[\vect{Y}]\neq t[\vect{Y}]}} R(u).
\]
Since by assumption $R(t[\vect{Y}])\neq 0$, we obtain by $\kplus$-positivity of $K$ that $R(t[\vect{Z}]) \neq 0$, i.e., $t[\vect{Z}] \in R[\vect{Z}]'$.
\end{proof}

\sound*
\begin{proof}
Triviality and Symmetry are clearly sound for all $K$-teams. We thus consider only Decomposition, Weak union, and Contraction.
Fix a $K$-relation $R(\vect{V})$ over some variable set $\vect{V}$ that contains $\vect{X}\vect{Y}\vect{Z}\vect{W}$.
\begin{itemize}
\item Decomposition: Suppose $R$ satisfies $\pci{\vect{X}}{\vect{Y}}{\vect{Z}\vect{W}}$. Then, for all tuples $t\in \tup(\vect{V})$ it holds that
\begin{equation}\label{eq:decom}
  R(t[\vect{X}\vect{Y}])   R(t[\vect{X}\vect{Z}\vect{W}]) =   R(t[\vect{X}])  R(t[\vect{X}\vect{Y}\vect{Z}\vect{W}])
\end{equation}
This implies 
\begin{align*}
& R(t[\vect{X}\vect{Y}])  R(t[\vect{X}\vect{Z}])   =  R(t[\vect{X}\vect{Y}])   \bigkplus_{\substack{t'\in \tup(\vect{X}\vect{Z}\vect{W})\\t'[\vect{X}\vect{Z}] =t[\vect{X}\vect{Z}]}} R(t'[\vect{X}\vect{Z}\vect{W}]) \\
=& \, \bigkplus_{\substack{t'\in \tup(\vect{X}\vect{Y}\vect{Z}\vect{W})\\t'[\vect{X}\vect{Y}\vect{Z}] =t[\vect{X}\vect{Y}\vect{Z}]}}   R(t'[\vect{X}\vect{Y}]) R(t'[\vect{X}\vect{Z}\vect{W}])
=    \bigkplus_{\substack{t'\in \tup(\vect{X}\vect{Y}\vect{Z}\vect{W})\\t'[\vect{X}\vect{Y}\vect{Z}] =t[\vect{X}\vect{Y}\vect{Z}]}} R(t'[\vect{X}]) R(t'[\vect{X}\vect{Y}\vect{Z}\vect{W}]) \\
=& \, R(t[\vect{X}])\bigkplus_{\substack{t'\in \tup(\vect{X}\vect{Y}\vect{Z}\vect{W})\\t'[\vect{X}\vect{Y}\vect{Z}] =t[\vect{X}\vect{Y}\vect{Z}]}} R(t'[\vect{X}\vect{Y}\vect{Z}\vect{W}]) 
=   R(t[\vect{X}])  R(t[\vect{X}\vect{Y}\vect{Z}]).
\end{align*}
Having showed
\begin{equation}\label{eq:conc}
R(t[\vect{X}\vect{Y}])  R(t[\vect{X}\vect{Z}]) = R(t[\vect{X}])  R(t[\vect{X}\vect{Y}\vect{Z}]),
\end{equation}
 we conclude that $R$ satisfies $\pci{\vect{X}}{\vect{Y}}{\vect{Z}}$.

\item Weak union: Suppose again $R$ satisfies $\pci{\vect{X}}{\vect{Y}}{\vect{Z}\vect{W}}$, in which case Eq. \eqref{eq:decom} holds for all tuples $t\in \tup(\vect{V})$. Multiplying both sides by $R[t(\vect{X}\vect{Z})] R[t(\vect{X}\vect{Y}\vect{Z})]$ yields
\begin{align*}
  &R(t[\vect{X}\vect{Y}])   R(t[\vect{X}\vect{Z}\vect{W}])  R[t(\vect{X}\vect{Z})] R[t(\vect{X}\vect{Y}\vect{Z})]\\
    =\,& R(t[\vect{X}])  R(t[\vect{X}\vect{Y}\vect{Z}\vect{W}]) R[t(\vect{X}\vect{Z})] R[t(\vect{X}\vect{Y}\vect{Z})].
\end{align*}
If $R(t[\vect{X}\vect{Y}])   R(t[\vect{X}\vect{Z}])\neq 0$, we may apply Eq. \eqref{eq:conc}, which is implied by Eq. \eqref{eq:decom}, and multiplicative cancellativity to obtain 
\begin{equation}\label{eq:obtain}
R(t[\vect{X}\vect{Z}\vect{W}])   R(t[\vect{X}\vect{Y}\vect{Z}]) =  R[t(\vect{X}\vect{Z})] R[t(\vect{X}\vect{Y}\vect{Z}\vect{W})].
\end{equation}
Suppose then $R(t[\vect{X}\vect{Y}])   R(t[\vect{X}\vect{Z}])= 0$. Since $K$ lacks zero divisors, either $R(t[\vect{X}\vect{Y}])=0$ or $R(t[\vect{X}\vect{Z}])=0$.
By positivity and and Lemma \ref{lem:poslemma}, it follows that  $R(t[\vect{X}\vect{Y}\vect{Z}])=0$  and $R(t[\vect{X}\vect{Y}\vect{Z}\vect{W}])=0$. 
 In particular, both sides of Eq. \eqref{eq:obtain} vanish. Hence we may conclude that $R$ satisfies $\pci{\vect{X}\vect{Z}}{\vect{Y}}{\vect{W}}$.

\item Contraction: Suppose $R$ satisfies $\pci{\vect{X}}{\vect{Y}}{\vect{Z}}$ and $\pci{\vect{X}\vect{Z}}{\vect{Y}}{\vect{W}}$, in which case we have Eq. \eqref{eq:obtain}  as well as
\begin{equation}\label{eq:contr-app}
R(t[\vect{X}\vect{Y}])  R(t[\vect{X}\vect{Z}]) = R(t[\vect{X}])  R(t[\vect{X}\vect{Y}\vect{Z}]),
\end{equation}
for all $t\in \tup(\vect{V})$.
Multiplying both left-hand sides and right-hand sides of Eqs. \eqref{eq:obtain} and \eqref{eq:contr-app} with one another yields
\begin{align*}
&R(t[\vect{X}\vect{Z}\vect{W}])   R(t[\vect{X}\vect{Y}\vect{Z}])  R(t[\vect{X}\vect{Y}])  R(t[\vect{X}\vect{Z}])  \\
=\,& R(t[\vect{X}\vect{Z}]) R(t[\vect{X}\vect{Y}\vect{Z}\vect{W}])  R(t[\vect{X}])  R(t[\vect{X}\vect{Y}\vect{Z}]).
\end{align*}
If $R(t[\vect{X}\vect{Z}])  R(t[\vect{X}\vect{Y}\vect{Z}]) \neq 0$, we obtain  Eq. \eqref{eq:decom} by multiplicative cancellativity. On the other hand, assuming $R(t[\vect{X}\vect{Z}])  R(t[\vect{X}\vect{Y}\vect{Z}]) = 0$ we have  three cases:
\begin{enumerate}
\item $R(t[\vect{X}\vect{Z}]) = R(t[\vect{X}\vect{Y}\vect{Z}])=0$. Then $R(t[\vect{X}\vect{Z}\vect{W}]) = R(t[\vect{X}\vect{Y}\vect{Z}\vect{W}])=0$ by positivity of $R$ and Lemma \ref{lem:poslemma}, wherefore both sides of   Eq. \eqref{eq:decom} vanish. 
\item $R(t[\vect{X}\vect{Z}]) = 0$ and  $R(t[\vect{X}\vect{Y}\vect{Z}])\neq 0$. Then $R(t[\vect{X}]) = 0$ by positivity of $K$ and Eq. \eqref{eq:contr-app}. Again, both sides of   Eq. \eqref{eq:decom} vanish.
\item $R(t[\vect{X}\vect{Z}]) \neq  0$ and  $R(t[\vect{X}\vect{Y}\vect{Z}])= 0$. This time $R(t[\vect{X}\vect{Y}]) = 0$ by positivity of $K$ and Eq. \eqref{eq:contr-app}, and once more we obtain  Eq. \eqref{eq:decom}.
\end{enumerate}
Since   Eq. \eqref{eq:decom} always holds, we conclude that  $R$ satisfies $\pci{\vect{X}}{\vect{Y}}{\vect{W}\vect{Z}}$.

\item Interaction: Suppose $R$ satisfies $\pci{\vect{X}\vect{W}}{\vect{Y}}{\vect{Z}}$ and $\pci{\vect{X}\vect{Z}}{\vect{Y}}{\vect{W}}$. 
Then, we have
\begin{align}
  R(t[\vect{X}\vect{Z}\vect{W}])  R(t[\vect{X}\vect{Y}\vect{W}]) &= R(t[\vect{X}\vect{W}])  R(t[\vect{X}\vect{Y}\vect{Z}\vect{W}]),\label{eq:rewritable}\\
  R(t[\vect{X}\vect{Y}\vect{Z}])  R(t[\vect{X}\vect{Z}\vect{W}]) &= R(t[\vect{X}\vect{Z}])  R(t[\vect{X}\vect{Y}\vect{Z}\vect{W}]),\label{eq:rewritable2}
\end{align}
for all tuples $t\in \tup(\vect{V})$. By Proposition \ref{prop:embedding}, $K$ embeds in some semifield $F$. 
Given $a \in K\setminus\{0\}$, let us write $a^{-1}$ for its multiplicative inverse in $F$.
By positivity of $K$, totality of $R$, and Lemma \ref{lem:poslemma}, we observe  $R(t[\vect{U}])\neq 0$ for all $\vect{U}\subseteq \vect{V}$ and $t \in \tup(\vect{V})$.
Thus, the expression $R(t[\vect{U}']\mid t([\vect{U}]) \coloneqq R(t[\vect{U}\vect{U}'])  R(t[\vect{U}])^{-1}$
is well defined for all $\vect{U},\vect{U}'\subseteq \vect{V}$ and $t \in \tup(\vect{V})$. We may now rewrite Eqs. \eqref{eq:rewritable} and \eqref{eq:rewritable2}
as
\[
  R(t[\vect{X}\vect{Y}\vect{W}]\mid t[\vect{X}\vect{W}])=R(t[\vect{X}\vect{Y}\vect{Z}\vect{W}]\mid t[\vect{X}\vect{Z}\vect{W}])=R(t[\vect{X}\vect{Y}\vect{Z}]\mid t[\vect{X}\vect{Z}])
\]
from which we obtain
\[
  R(t[\vect{X}\vect{Y}\vect{W}])  R(t[\vect{X}\vect{Z}])=R(t[\vect{X}\vect{Y}\vect{Z}]) R(t[\vect{X}\vect{W}])
\]
for all $t \in \tup(\vect{V})$. Thus, for arbitrary $t \in \tup(\vect{V})$, 
\begin{align*}
  R(t[\vect{X}\vect{Y}])  R(t[\vect{X}\vect{Z}])=&R(t[\vect{X}\vect{Z}])\bigkplus_{\substack{t'\in \tup(\vect{XYW})\\t'[\vect{XY}]=t[\vect{XY}]}} R(t'[\vect{XY}\vect{W}])  \\
  =&\bigkplus_{\substack{t'\in \tup(\vect{XYZW})\\t'[\vect{XYZ}]=t[\vect{XYZ}]}}R(t'[\vect{X}\vect{Z}]) R(t'[\vect{XY}\vect{W}])\\
  =&\bigkplus_{\substack{t'\in \tup(\vect{XYZW})\\t'[\vect{XYZ}]=t[\vect{XYZ}]}}R(t'[\vect{X}\vect{Y}\vect{Z}]) R(t'[\vect{X}\vect{W}])\\ 
    =&R(t[\vect{X}\vect{Y}\vect{Z}])\bigkplus_{\substack{t'\in \tup(\vect{XW})\\t'[\vect{X}]=t[\vect{X}]}} R(t'[\vect{X}\vect{W}])\\ 
  =&R(t[\vect{X}\vect{Y}\vect{Z}]) R(t[\vect{X}]),
\end{align*}
by which we conclude that $R$ satisfies $\pci{\vect{X}}{\vect{Y}}{\vect{Z}}$. Hence, $R$ also satisfies $\pci{\vect{X}}{\vect{Y}}{\vect{Z}\vect{W}}$
by soundness of Contraction.
\end{itemize}
\end{proof}

\section{Combined rules}\label{sect:combined}

We first show the following helping lemma.

\begin{lemma}\label{lem:notzero}
Let $R(\vect{X})$ be a $K$-relation, and let $\vect{Y},\vect{Z}\subseteq \vect{X}$. If
$R\models \vect{Y} \to \vect{Z}$, then $R(t[\vect{Y}\vect{Z}])\neq 0$ implies $R(t[\vect{Y}])=R(t[\vect{Y}\vect{Z}])$.
\end{lemma}
\begin{proof}
Applying Eq. \eqref{eq:marg}  and $R(t[\vect{Y}\vect{Z}])\neq 0$, we find $t'\in  \tup(\vect{X})$ such that $t'[\vect{Y}\vect{Z}]=t[\vect{Y}\vect{Z}]$ and $R(t')\neq 0$. 
On the other hand, since $R'\models \vect{Y} \to \vect{Z}$, for all $t''\in \tup(\vect{X})$ we have $R(t'')= 0$ whenever $t''[\vect{Y}]=t'[\vect{Y}]$ and
$t''[\vect{Z}]\neq t'[\vect{Z}]$. Thus, applying Eq. \eqref{eq:marg} again, for all $s\in \tup(\vect{Y}\vect{Z})$ we have $R(s)= 0$ whenever $s[\vect{Y}]=t[\vect{Y}]$ and
$s[\vect{Z}]\neq t[\vect{Z}]$.

Then, since $R[\vect{Y}]= R[\vect{Y}\vect{Z}][\vect{Y}]$ by Lemma \ref{lem:phokion}, we may conclude that
\[
R(t[\vect{Y}]) = R[\vect{Y}\vect{Z}](t[\vect{Y}])=\bigkplus_{\substack{t'\in \tup(\vect{Y}\vect{Z})\\t'[\vect{Y}]=t[\vect{Y}]}} R(t')
= R(t[\vect{Y}\vect{Z}]) + \overbrace{\bigkplus_{\substack{s\in \tup(\vect{Y}\vect{Z})\\s[\vect{Y}]=t[\vect{Y}]\\s[\vect{Z}]\neq t[\vect{Z}]}} R(s)}^{0} = R(t[\vect{Y}\vect{Z}]).
\]
\end{proof}

\cifd*
\begin{proof}
Let $K$ be a semiring, and let $R(\vect{V})$ be a $K$-relation, where $\vect{X},\vect{Y},\vect{Z}\subseteq \vect{V}$ are disjoint sets of variables.
Considering the first statement, assume that $R$ satisfies $\vect{X} \to \vect{Y}$.  
We need to show that $R$ satisfies $\pci{\vect{X}}{\vect{Y}}{\vect{Z}}$, i.e.,
\begin{equation}\label{eq:fdci}
R(t[\vect{X}\vect{Y}])  R(t[\vect{X}\vect{Z}]) = R(t)  R(t[\vect{X}])
\end{equation}
 for all $t \in \tup(\vect{X}\vect{Y}\vect{Z})$. Assume first that $R(t[\vect{X}\vect{Y}]) \neq 0 \neq R(t[\vect{X}\vect{Z}])$.
 By Eq. \eqref{eq:marg}, there must be $t',t''\in \tup(\vect{X}\vect{Y}\vect{Z})$ such that $R(t')\neq 0\neq R(t'')$, $t'[\vect{X}\vect{Y}]=t[\vect{X}\vect{Y}]$, and $t''[\vect{X}\vect{Z}]=t[\vect{X}\vect{Z}]$. Since $R$ satisfies $\vect{X} \to \vect{Y}$, it must also be the case that $t'[\vect{Y}]=t''[\vect{Y}]$, wherefore we obtain $t=t''$. In particular, $R(t)\neq 0$.
  
Since $\vect{X} \to \vect{Y}$ entails $\vect{X}\vect{Z} \to \vect{Y}\vect{Z}$ by the augmentation rule, we can now apply 
Lemma \ref{lem:notzero} to obtain $R(t)=R(t[\vect{X}\vect{Z}]) $ and $R(t[\vect{X}\vect{Y}])= R(t[\vect{X}])$. 
This proves
Eq. \eqref{eq:fdci}.

Let us then assume that $R(t[\vect{X}\vect{Y}]) = 0$ or $R(t[\vect{X}\vect{Z}])=0$, i.e., $t[\vect{X}\vect{Y}] \notin R[\vect{X}\vect{Y}]'$ or $t[\vect{X}\vect{Z}] \notin R[\vect{X}\vect{Z}]'$. By Lemma \ref{lem:poslemma}, this leads to $t \notin R'$, i.e., $R(t)=0$. Consequently, Eq. \eqref{eq:fdci} holds true since both its left-hand side and right-hand side are zero.
We conclude that  CI introduction is sound for all $K$-relations, where $K$ is $\kplus$-positive.

Let us then consider the second statement. 
To show the contraposition, assume that $R$ does not satisfy $\vect{X} \to \vect{Z}$. Then we find two tuples $t$ and $t'$ from $R'$ such that
$t[\vect{X}]=t'[\vect{X}]$ and $t[\vect{Z}]\neq t'[\vect{Z}]$. 
By Lemma \ref{lem:poslemma} we obtain $t[\vect{X}\vect{Y}] \in R[\vect{X}\vect{Y}]'$ and $t'[\vect{X}\vect{Z}] \in R[\vect{X}\vect{Z}]'$, i.e., $R(t[\vect{X}\vect{Y}] )\neq 0$ and $R(t'[\vect{X}\vect{Z}] )\neq 0$.
Fix now a $\vect{X}\vect{Y}\vect{Z}$-tuple $t''$ that agrees with $t$ on $\vect{X}\vect{Y}$ and with $t'$ on $\vect{Z}$ (and consequently on $\vect{X}$ as well). Assume now that $R$ satisfies $\pci{\vect{X}}{\vect{Y}}{\vect{Z}}$, in which case
\[
R(t''[\vect{X}\vect{Y}])R(t''[\vect{X}\vect{Z}]) = R(t'')R(t''[\vect{X}]).
\]
Since $K$ does not have zero divisors, the left-hand side of this equation is not zero, meaning that $R(t'')$ is not zero. By Eq. \eqref{eq:marg} we can extend $t''$ to a $\vect{V}$-tuple $t^*$ such that $R(t^*)\neq 0$. This means that the support $R'$ contains two tuples, $t$ and $t^*$, which agree on $\vect{X}\vect{Y}$ but disagree on $\vect{Z}$. In particular, $R$ does not satisfy the FD $\vect{X}\vect{Y}\to \vect{Z}$. This shows the contraposition of the second statement.
\end{proof}

\section{Axioms of SCI+FD}\label{sect:axiomssci}

First we define explicitly the operation that connects expressions from probability theory with their counterparts in database theory.
 For an MVD $\sigma = \vect{X} \twoheadrightarrow \vect{Y}$, the \emph{corresponding SCI} is of the form
$\pci{\vect{X}}{\vect{Y}\setminus \vect{X}}{\vect{V}\setminus \vect{X}\vect{Y}}$. Conversely, for an SCI $\pci{\vect{X}}{\vect{Y}}{\vect{Z}}$, the \emph{corresponding MVD} is of the form $\vect{X} \twoheadrightarrow \vect{Y}$. Analogously, for an EMVD $\sigma = \vect{X} \twoheadrightarrow \vect{Y}\mid \vect{Z}$, the \emph{corresponding CI} is the SCI that corresponds to the MVD $\vect{X} \twoheadrightarrow \vect{Y}$ on the schema $\vect{X}\vect{Y}\vect{Z}$. For an CI, the \emph{corresponding EMVD} is defined in the obvious way. 
We then let $\sigma \mapsto \sigma^*$ be the operation that sends $\sigma$ to its corresponding statement $\sigma^*$, and extend this operation to be the identity on FDs. The operation is also extended to sets by $\Sigma^* \coloneqq \{\sigma^* \mid \sigma \in \Sigma\}$. 
 A relation $R$ satisfies an MVD/EMVD/FD $\sigma$ if and only if its characteristic $\mathbb{B}$-relation satisfies $\sigma^*$.

  \begin{figure}[h!]
  \centering
  \begin{tikzpicture}[every node/.style={outer sep=0pt}]
  \def\m{1.4em}
    \node[draw,minimum width=0.5\textwidth,rounded corners, minimum height=3.2cm] (box1) {
    \hspace{15mm}
      \begin{minipage}[t][2.7cm]{.87\textwidth}
 \begin{itemize}
 \item[MVD0] Complementation: If $\vect{XYZ}=\vect{V}$, $\vect{Y}\cap\vect{Z}\subseteq\vect{X}$, and $\vect{X}\twoheadrightarrow \vect{Y}$, then $\vect{X}\twoheadrightarrow \vect{Z}$. 
 \item[MVD1] Reflexivity: If $\vect{Y}\subseteq \vect{X}$, then $\vect{X}\twoheadrightarrow \vect{Y}$. 
  \item[MVD2]  Augmentation: If $\vect{Z}\subseteq \vect{W}$ and $\vect{X}\twoheadrightarrow \vect{Y}$, then $\vect{XW}\twoheadrightarrow \vect{Y}\vect{Z}$. 
 \item[MVD3] Transitivity: If $\vect{X}\twoheadrightarrow \vect{Y}$ and $\vect{Y}\twoheadrightarrow \vect{Z}$, then $\vect{X}\twoheadrightarrow \vect{Z}\setminus \vect{Y}$. 
 \item[MVD-FD1] If $\vect{X}\twoheadrightarrow \vect{Y}$, then $\vect{X}\to \vect{Y}$.
  \item[MVD-FD2] If $\vect{Z}'\subseteq \vect{Z}$, $\vect{Y}\cap\vect{Z}=\emptyset$, $\vect{X}\twoheadrightarrow \vect{Z}$, and $\vect{Y}\to \vect{Z}'$, then $\vect{X}\to \vect{Z}'$.
 \end{itemize}
  \end{minipage}
    };
  \end{tikzpicture}
\caption{Axioms for MVDs and FDs. \label{fig:mvdax}}
\end{figure}

 \begin{theorem}[\cite{BeeriFH77}]\label{thm:beeri}
 The axioms of MVDs and FDs (Fig. \ref{fig:mvdax}) are sound and complete.
 \end{theorem}

Since $(\sigma^*)^*=\sigma$ for any CI/FD, the following lemma entails the direction (4) $\Rightarrow$ (1) of Theorem \ref{thm:kenig}.
\begin{lemma}\label{lem:comp}
Let $\Sigma\cup\{\tau\}$ be a set of MVDs and FDs over a variable set $\vect{V}$. 
If $\Sigma$ implies $\tau$ over relations, then $\tau^*$ is can be derived from $\Sigma^*$ using (S1-S5), (FD1-FD3), and (FD-CI1,FD-CI2). 
\end{lemma}
%
 %
\begin{proof}
 Assume now that $\tau$ is provable from $\Sigma$ by the axioms for MVDs and FDs. We claim that $\tau^*$ is provable from $\Sigma^*$ by using (S1-S5), (FD1-FD3), and (FD-CI1,FD-CI2). 
 The proof is by induction on the length of the proof. 
 If $\tau$ is obtained from $\sigma_1, \dots ,\sigma_n$ by applying one the rules for MVDs and FDs, the induction hypothesis is that $\sigma^*_1, \dots ,\sigma^*_n$ are provable from $\Sigma$ by (S1-S5), (FD1-FD3), and (FD-CI1,FD-CI2).
  The case where $\tau$ is obtained by one of the Armstrong axioms is clear by induction hypothesis. The remaining cases are considered below. (The semigraphoid axioms are known to be sound and complete for SCIs interpreted as MVDs on database relations \cite{GyssensNG14}. Thus, the cases for the axioms (MVD0-MVD3) follow readily. We however consider all the remaining cases for the sake of completeness.)
  \begin{itemize}
  \item MVD0 (Complementation): Suppose $\tau=\vect{X}\twoheadrightarrow \vect{Z}$ is deduced from $\sigma=\vect{X}\twoheadrightarrow \vect{Y}$, where $\vect{XYZ}=\vect{V}$, $\vect{Y}\cap\vect{Z}\subseteq\vect{X}$. Then $\tau^*=\pci{\vect{X}}{\vect{Z}\setminus \vect{X}}{\vect{V}\setminus \vect{XZ}}$ and $\sigma^*=\pci{\vect{X}}{\vect{Y}\setminus \vect{X}}{\vect{V}\setminus \vect{XY}}$. From $\sigma^*$ we obtain $\pci{\vect{X}}{\vect{Y}\setminus \vect{X}}{\vect{Z}\setminus \vect{X}}$ by decomposition. Then, $\tau^*$ follows by symmetry, and since $\vect{V}\setminus \vect{XZ}=\vect{Y}\setminus \vect{X}$ by hypothesis.
 \item MVD1 (Reflexivity): If $\tau$ is of the form $\vect{X}\twoheadrightarrow \vect{Y}$, where $\vect{Y}\subseteq \vect{X}$, then $\tau^*$ can be deduced by the symmetry and triviality rules of the semigraphoid axioms.
  \item MVD2 (Augmentation): Suppose $\tau=\vect{XW}\twoheadrightarrow \vect{Y}\vect{Z}$ is obtained from $\vect{Z}\subseteq \vect{W}$ and $\vect{X}\twoheadrightarrow \vect{Y}$. Since $\tau^* = (\vect{Z}\subseteq \vect{W})^*$, the claim follows by induction hypothesis.
 \item MVD3 (Transitivity): Suppose $\tau=\vect{X}\twoheadrightarrow \vect{Z}\setminus \vect{Y}$ is obtained from $\sigma_0=\vect{X}\twoheadrightarrow \vect{Y}$ and $\sigma_1=\vect{Y}\twoheadrightarrow \vect{Z}$. Let us write 
 \begin{itemize}
 \item $\vect{A}= \vect{X}\setminus \vect{YZ}$,
  \item $\vect{B}= \vect{Y}\setminus \vect{XZ}$,
 \item $\vect{C}= \vect{Z}\setminus \vect{XY}$,
\item $\vect{D}=(\vect{X}\cap\vect{Y})\setminus \vect{Z}$,
 \item $\vect{E}=(\vect{X}\cap\vect{Z})\setminus \vect{Y}$,
\item $\vect{F}=(\vect{Y}\cap\vect{Z})\setminus \vect{X}$,
\item $\vect{G}=\vect{X}\cap\vect{Y}\cap \vect{Z}$,
\item $\vect{H}=\vect{V} \setminus \vect{XYZ}$.
 \end{itemize}
 We need to prove $\tau^*=\pci{\vect{ADEG}}{\vect{C}}{\vect{BFH}}$ from $\sigma_0^*=\pci{\vect{ADEG}}{\vect{BF}}{\vect{CH}}$ and $\sigma_1^*=\pci{\vect{BDFG}}{\vect{CE}}{\vect{AH}}$. This can be done as follows:
 \begin{enumerate}
 \item $\pci{\vect{ABDEG}}{\vect{CH}}{\vect{F}}$ (from $\sigma_1^*$ using symmetry and weak union)
  \item $\pci{\vect{ABDEG}}{\vect{C}}{\vect{F}}$ (from (1) by decomposition)
  \item $\pci{\vect{ABDEFG}}{\vect{C}}{\vect{H}}$ (from $\sigma_2^*$ using symmetry and weak union)
   \item $\pci{\vect{ABDEG}}{\vect{C}}{\vect{FH}}$ (from (2) and (3) by contraction)
   \item $\pci{\vect{ADEG}}{\vect{C}}{\vect{B}}$ (from $\sigma_1^*$ using symmetry and decomposition)
   \item $\pci{\vect{ADEG}}{\vect{C}}{\vect{BFH}}$ (from (4) and (5) by contraction)
 \end{enumerate}
 \item MVD-FD1: Suppose $\tau = \vect{X} \twoheadrightarrow \vect{Y}$ is obtained from $\sigma = \vect{X}\to \vect{Y}$. We can deduce $\vect{X} \to \vect{Y}\setminus \vect{X}$ from $\sigma^*$ using transitivity and reflexivity of the Armstrong axioms. Then, by  CI introduction we obtain $\sigma^*=\pci{\vect{X}}{\vect{Y}\setminus \vect{X}}{\vect{V}\setminus \vect{XY}}$.
  \item MVD-FD2: Suppose $\tau=\vect{X}\to \vect{Z}'$ is obtained from $\sigma_0=\vect{X}\twoheadrightarrow \vect{Z}$ and
  $\sigma_1= \vect{Y}\to \vect{Z}'$, where $\vect{Z}' \subseteq \vect{Z}$ and $\vect{Y}\cap\vect{Z}=\emptyset$.
  We need to show that $\tau^*=\tau$ can be deduced from $\sigma^*_0=\pci{\vect{X}}{\vect{Z}\setminus \vect{X}}{\vect{V}\setminus \vect{XZ}}$ and
  $\sigma^*_1= \sigma_1$. Using symmetry and decomposition we obtain $\pci{\vect{X}}{\vect{Y}\setminus \vect{X}}{\vect{Z}'\setminus \vect{X}}$ from $\sigma^*_0$. Using reflexivity and transitivity $\vect{X}\vect{Y}\to \vect{Z}'$ follows from $\sigma_1^*$. By FD contraction, we then deduce $\vect{X}\to \vect{Z}'\setminus \vect{X}$ from $\pci{\vect{X}}{\vect{Y}\setminus \vect{X}}{\vect{Z}'\setminus \vect{X}}$ and $\vect{X}\vect{Y}\to \vect{Z}'$. Finally, we obtain $\vect{X}\to \vect{Z}'$ from $\vect{X}\to \vect{Z}'\setminus \vect{X}$ by reflexivity and transitivity. 
 \end{itemize}
This completes the induction proof. Thus, the statement of the lemma follows.
\end{proof}


\end{document}